\documentclass[%
 reprint,
 amsmath,amssymb,
 aps,
 nofootinbib,
 superscriptaddress
]{article}

\usepackage{amsthm}
\usepackage{graphicx}
\usepackage{xcolor}
\usepackage{braket}
\usepackage{amsmath}
\usepackage{amssymb}
\usepackage{enumerate}
\usepackage{booktabs}
\usepackage{txfonts}
\usepackage[subrefformat=parens]{subcaption}
\usepackage{caption}
\usepackage{authblk}
\usepackage{cite}
\usepackage{mathtools}
\usepackage[a4paper]{geometry}
\usepackage{algorithmic}
\usepackage{algorithm}
\usepackage{bbm}
\usepackage{eqparbox}
\usepackage{fancyhdr}

\newtheorem{theorem}{Theorem}
\newtheorem{lemma}{Lemma}

\newtheorem{assum}{Assumption}
\newtheorem{problem}{Problem}

\newcommand{\proc}[1]{\textup{\textsf{#1}}}

\begin{document}
\fancypagestyle{plain}{\fancyhead[R]{IFT-UAM/CSIC-22-44}}
\renewcommand{\headrulewidth}{0pt}

\title{Gravitational wave matched filtering by quantum Monte Carlo integration and quantum amplitude amplification}
\renewcommand{\thefootnote}{\fnsymbol{footnote}}

\author{
Koichi Miyamoto$^{1}$\footnote{miyamoto.kouichi.qiqb@osaka-u.ac.jp}, 
Gonzalo Morr\'{a}s$^{2}$\footnote{gonzalo.morras@estudiante.uam.es},
Takahiro S. Yamamoto$^{3}$\footnote{yamamoto.takahiro@f.mbox.nagoya-u.ac.jp},
Sachiko Kuroyanagi$^{2,3}$\footnote{sachiko.kuroyanagi@csic.es} \ and
Savvas Nesseris$^{2}$\footnote{savvas.nesseris@csic.es}
\\

\footnotesize $^{1}$ Center for Quantum Information and Quantum Biology, Osaka University, Toyonaka, 560-8531, Japan \\
\footnotesize $^{2}$ Instituto de F\'isica Te\'orica UAM-CSIC, Universidad Auton\'oma de Madrid, Cantoblanco, 28049 Madrid, Spain\\
\footnotesize $^{3}$ Department of Physics and Astrophysics, Nagoya University, Nagoya, 464-8602, Japan\\	
}

\date{\today}
\maketitle

\renewcommand{\thefootnote}{\arabic{footnote}}

\begin{abstract}

The speedup of heavy numerical tasks by quantum computing is now actively investigated in various fields including data analysis in physics and astronomy.
In this paper, we propose a new quantum algorithm for matched filtering in gravitational wave (GW) data analysis based on the previous work by Gao et al., Phys. Rev. Research 4, 023006 (2022) [arXiv:2109.01535].
Our approach uses the quantum algorithm for Monte Carlo integration for the signal-to-noise ratio (SNR) calculation instead of the fast Fourier transform used in Gao et al. and searches signal templates with high SNR by quantum amplitude amplification.
In this way, we achieve an exponential reduction of the qubit number compared with Gao et al.'s algorithm, keeping a quadratic speedup over classical GW matched filtering with respect to the template number.

\end{abstract}

\section{Introduction}

Quantum computing \cite{nielsen2002} is a developing technology and is expected to speed up some classes of computation that are intractable in classical computing.
The recent rapid advance in quantum computer development has been stimulating research on applications of quantum algorithms to concrete problems in various fields (see Ref.~\cite{lin2022} for a recent review).

In this paper, we study an application of some quantum algorithms to a problem in gravitational wave (GW) experiments, which measures the space-time distortion caused by GWs with a laser interferometer.
A recent paper \cite{gao2021} has proposed the use of Grover’s search algorithm for matched filtering, which is a commonly used technique to search a signal buried in noisy data and is widely used in GW data analysis \cite{balasubramanian1996,owen1996,owen1999,allen2012}.
Given a target signal waveform, called a template, we take an inner product between a template and data to cancel out the noise contribution and extract a signal.
In fact, the first GW was detected by the LIGO detectors in 2015 \cite{abbott2016}, and, after that, the worldwide GW detector network has observed tens of GW events using matched filtering \cite{LIGOScientific:2018mvr,abbott2021,LIGOScientific:2021djp}.

A challenging point in GW matched filtering is the large number of templates.
The functional form of a GW signal is predicted by general relativity depending on the GW source such as compact binary
coalescence \cite{balasubramanian1996,owen1996,owen1999,allen2012}, but it has some parameters such as masses of the compact objects, spin parameters, and luminosity distance.
In order to get a high signal-to-noise ratio (SNR), we must perform matched filtering using a template with appropriate parameters.
Therefore, usually, we set sufficiently many points in the parameter space 
and run an exhaustive search.
That is, we repeat matched filtering using templates one by one to find those that yield SNR larger than some threshold.
This is an extremely time-consuming task and expected to be sped up by quantum computing.

Fortunately, there exists a quantum algorithm for searching, called Grover's algorithm \cite{grover1996}.
Given $N$ data, $x_1,...,x_N$, represented as bit strings and a condition $F$ as a function that maps a bit string to 0 or 1, Grover's algorithm can find ``marked data" $x$ such that $F(x)=1$ making $O(\sqrt{N/n})$ calls to $F$, where $n$ is the number of marked data.
Therefore, it is often said that Grover's algorithm provides a quadratic speedup over the classical exhaustive search, which has $O(N/n)$ query complexity.
As an application of this, Ref.~\cite{gao2021} presented a quantum algorithm for GW matched filtering.
With the SNR calculation implemented as a quantum circuit, the aforementioned algorithm can find a template with SNR higher than a threshold $\rho_{\rm th}$ with $\widetilde{O}\left(M/\sqrt{r(\rho_{\rm th})}\right)$ complexity\footnote{In the big-O notation, we use a symbol $\widetilde{O}(\cdot)$, which hides logarithmic factors in $O(\cdot)$.}, where $r(\rho_{\rm th})$ is the fraction of templates that yield SNR $\rho\ge\rho_{\rm th}$ and $M$ is the number of points in the time-series data of the detector output, or, equivalently, the number of frequency bins of Fourier transformed data.
This is in fact a quadratic speedup over the classical method, which has $O(M/r(\rho_{\rm th}))$ complexity, with respect to the template number.

However, this quantum algorithm has the following subtlety.
It uses Fast Fourier Transform (FFT) \cite{cooley1967} for the SNR calculation, which is also used in the usual classical way.
FFT simultaneously calculates SNR for $M$ possible values of a parameter called time of coalescence, with other parameter fixed, in $O(M\log M)$ time, whereas naively such a computation takes $O(M^2)$ time without FFT.
However, in order to store the intermediate and final calculation results, FFT requires $O(M)$ qubits, which is a somewhat large number since $M$ is typically of order $4096{\rm Hz} \times 256{\rm s} \sim O(10^6)$, where $4096$Hz is the sampling frequency and $256$s is the typical duration of the data segment~\cite{allen2012}.
This might cause an issue on feasibility, since fault-tolerant quantum computers will have a limitation on the number of qubits available even in the future.
It is expected that creating one logical qubit requires thousands or tens of thousands physical qubits for error correction \cite{fowler2012}, and therefore realizing a quantum computer with millions of qubits is very challenging.

In light of this, we propose an alternative quantum algorithm for GW matched filtering, in which the SNR calculation with FFT is replaced with the quantum algorithm for Monte Carlo integration (QMCI) \cite{montanaro2015}.
QMCI is a method to estimate an integral given in the finite sum approximation.
Thus, it can be applied to the calculation of SNR, which includes frequency band integration and is in practice calculated as a sum of contributions from many Fourier modes.
In this approach, the required qubit number scales on $M$ as $O({\rm poly}(\log M))$, which means an exponential reduction from the FFT approach.
Note that this is not just a straightforward application of another quantum algorithm to a part of an existing method, since the use of QMCI causes the following issue.
Unlike FFT, which calculates SNR deterministically, QMCI inevitably accompanies errors, and thus comparing the SNR calculated by QMCI with a single SNR threshold $\rho_{\rm th}$ leads to a false alarm that the detector output yields SNR larger than $\rho_{\rm th}$ for some templates despite there being no such event.
As a solution to this, we propose to set two thresholds $\rho_{\rm hard}$ and $\rho_{\rm soft}$ that have the following meanings: we should never miss events with SNR $\rho\ge\rho_{\rm hard}$, and we do not want to be falsely alarmed by events with $\rho<\rho_{\rm soft}$.
Then, with QMCI accuracy set according to the difference between $\rho_{\rm hard}$ and $\rho_{\rm soft}$, the proposed algorithm says ``there is a signal" for events with SNR $\rho\ge\rho_{\rm hard}$ with high probability, ``there is no signal" for events with SNR $\rho<\rho_{\rm soft}$ with certainty, and either of these messages for events with SNR $\rho\in[\rho_{\rm soft}, \rho_{\rm hard})$.
The query complexity in this algorithm is of order $\widetilde{O}\left(M/\sqrt{r(\rho_{\rm hard})}\right)$, which still indicates a quadratic speedup.

The remaining part of this paper is organized as follows.
Section \ref{sec:preliminary} introduces the preliminary knowledge.
We outline GW matched filtering and some building-block quantum algorithms such as Grover's algorithm, quantum amplitude amplification (QAA), quantum amplitude estimation (QAE), and QMCI.
Section \ref{sec:main} is the main part.
Defining GW matched filtering as a mathematical problem, we explain the existing algorithm in \cite{gao2021}, and present our modified algorithm in detail, along with estimation of the query complexity and the qubit number and a plausible setting on thresholds $\rho_{\rm hard}$ and $\rho_{\rm soft}$.
Section \ref{sec:summary} summarizes this paper.
Some proofs are presented in appendices.

\section{Preliminary \label{sec:preliminary}}

\subsection{Notation}

Here, we summarize some notations used in this paper.
$\mathbb{R}_+$ denotes the set of all positive real numbers. 
For $n\in\mathbb{N}$, we define $[n]:=\{1,...,n\}$ and $[n]_0:=\{0,1,...,n-1\}$.
For any $x\in\mathbb{R}$, if $|x-y|\le\epsilon$ holds for some $y\in\mathbb{R}$ and $\epsilon\in\mathbb{R}_+$, we say that $x$ is $\epsilon$-close to $y$ and that $x$ is an $\epsilon$-approximation of $y$.
For any equation or inequality $C$, $\mathbbm{1}_C$ takes 1 if $C$ is satisfied, and 0 otherwise.
For $c\in\mathbb{C}$, $c^*$ denotes its complex conjugate.
For $\mathcal{X}=\{x_1,...,x_n\}$ and $\mathcal{Y}=\{y_1,...,y_n\}$, finite sets of real numbers  with same size $n$, we define the sample mean ${\rm Mean}(\mathcal{X}):=\frac{1}{n}\sum_{i=1}^nx_i$, the sample variance ${\rm Var}(\mathcal{X}):=\frac{1}{n}\sum_{i=1}^n\left(x_i-{\rm Mean}(\mathcal{X})\right)^2$ and the sample covariance ${\rm Cov}(\mathcal{X},\mathcal{Y}):=\frac{1}{n}\sum_{i=1}^n\left(x_i-{\rm Mean}(\mathcal{X})\right)\left(y_i-{\rm Mean}(\mathcal{Y})\right)$.
For $n\in\mathbb{N}$, $\mathbb{I}_n$ denotes the $n\times n$ identity matrix.
For $z\in\mathbb{C}$, $\Re z$ and $\Im z$ are the real and imaginary parts of $z$.

\subsection{Gravitational wave matched filtering}

Here, we outline matched filtering in GW search experiments \cite{balasubramanian1996,owen1996,owen1999,allen2012}.
Suppose that we are given the detector output $s(t)$ as a function of time $t$, which is a sum of the signal $h(t)$ and the noise $n(t)$:
\begin{equation}
    s(t)=h(t)+n(t).
\end{equation}
We assume that the noise is Gaussian, which means that, for each $f\in\mathbb{R}_+$, $\Re\tilde{n}(f)$ and $\Im\tilde{n}(f)$ are normal random variables and
\begin{equation}
    \mathbb{E}_{\rm n}[\tilde{n}(f)\tilde{n}^*(f^\prime)]=\frac{1}{2}S_{\rm n}(|f|)\delta(f-f^\prime) \label{eq:noisePScont}
\end{equation}
holds with the single-sided power spectrum density (PSD) $S_{\rm n}$.
Here, for any function $q(t)$ in time domain, $\tilde{q}(f):=\int^{\infty}_{-\infty}dt e^{2\pi i ft}q(t)$ is its Fourier transform, $\mathbb{E}_{\rm n}[\cdot]$ denotes an expectation with respect to randomness of the noise and $\delta(\cdot)$ is the Dirac delta function.
We define the inner product of two functions $q(t)$ and $q^\prime(t)$ in time domain as
\begin{equation}
    (q|q^\prime):=4\Re\left(\int^\infty_0df \frac{\tilde{q}^*(f)\tilde{q}^\prime(f)}{S_{\rm n}(f)}\right).
\end{equation}
The matched filtering search is peformed by taking an inner product of $s(t)$ and an appropriate filter function $Q(t)$ that yields a large inner product with the targeted signal.
The function $Q(t)$ is often called a template and normalized as $(Q|Q)=1$.
The template bank, the collection of templates, is prepared based on theoretically predicted waveform of signals.
The SNR for a detector output $s$ and template $Q$ is then defined as
\begin{equation}
    \rho = \frac{(Q|s)}{\sqrt{\mathbb{E}_{\rm n}[|(Q|n)|^2]}}=4\Re\left(\int^\infty_0df \frac{\tilde{Q}^*(f)\tilde{s}(f)}{S_{\rm n}(f)}\right).
\end{equation}
By setting a SNR threshold $\rho_{\rm th}$, exhaustive search is performed to find a template that gives a SNR larger than $\rho_{\rm th}$ to claim a detection.
In reality, non stationary detector noise known as glitches can generate large values of $\rho$ and cause false detections. To mitigate the effect of glitches, what is done in practice is to modify the variable used to rank the events. Instead of the SNR, real GW searches use more complicated ranking statistics which include signal consistency tests such as $\chi^2$ \cite{Allen:2004gu}. In this paper, we will only consider the SNR as the ranking statistic for simplicity, but the ideas of our quantum algorithm could be applied to searches with more complicated ranking statistics.
Hereafter, we write each Fourier transformed template as $\tilde{Q}_m(f)e^{-2\pi i ft_0}$, where $t_0$ is 
the time of coalescence and the dependency on other parameters (intrinsic parameters) is put into $\tilde{Q}_m$.
Here, we assume that there are $N_{\rm temp}$ candidates of the intrinsic parameter set in the template bank and label the functions $\tilde{Q}_m$ by $m\in[N_{\rm temp}]_0$.

In reality, we have a detector output as a sequence of discrete points in time \cite{allen2012}.
Suppose that a sequence of $s(t)$ is given by $\{s(\tau_l)\}_{l=0,...,M-1}$ at $M$ time points $\tau_0=0,\tau_1=\Delta t,...,\tau_{M-1}=(M-1)\Delta t$ with interval $\Delta t$.
In such a situation, Fourier transforms are given in the discrete form: for each $k\in \left\{0,1,...,M-1\right\}$. 
Thus, we redefine $\tilde{s}$ as
\begin{equation}
    \tilde{s}(f_k):=\Delta t \sum_{l=0}^{M-1} s(\tau_l) e^{2\pi ikl/M},
\end{equation}
where $f_k=k/T$ and $T=M\Delta t$, and similar quantities $\tilde{h}(f_k)$ and $\tilde{n}(f_k)$ for $h(t)$ and $n(t)$.
$\Re\tilde{n}(f_k)$ and $\Im\tilde{n}(f_k)$ are still normal random variables but Eq.~(\ref{eq:noisePScont}) is now converted into \cite{talbot2020}
\begin{equation}
    \mathbb{E}_{\rm n}[\tilde{n}(f_k)\tilde{n}^*(f_l)]=\frac{1}{2}S_{\rm n}(f_k)T\delta_{k,l}, \label{eq:noisePSdisc}
\end{equation}
where $\delta_{k,l}$ is the Kronecker delta.
Then, the SNR now becomes
\begin{equation}
	\rho_{m,j}=\frac{4}{M\Delta t}\Re\left(\sum^{\frac{M}{2}-1}_{k=1} \frac{\tilde{Q}^*_m(f_k)\tilde{s}(f_k)}{S_{\rm n}(f_k)}e^{2\pi i jk/M}\right) \,,
	\label{eq:SNRdiscorg}
\end{equation}
for the $m$-th intrinsic parameter set and the time of coalescence given as $t_0=j\Delta t$ with $j\in[M]_0$ \footnote{The sum in Eq.~(\ref{eq:SNRdiscorg}) runs over $k\in\left[\frac{M}{2}-1\right]$ rather than $k\in[M]_0$, because we use the one-sided power spectral density as defined in Eq.~(\ref{eq:noisePScont}). See \cite{allen2012} for the detail.}.
Here and hereafter, we assume that $M$ is even.
In the usual way of classical computing, although it seemingly takes $O(M^2)$ computational time to compute Eq.~(\ref{eq:SNRdiscorg}) for all $j\in[M]_0$ with $m$ fixed, we can do this in $O(M\log M)$ time using FFT \cite{cooley1967}.
This means that, for a fixed intrinsic parameter set, we can quickly search the optimal $t_0$ in $\{j\Delta t\}_{j=0,...,M-1}$ and obtain the optimal SNR $\rho_m:=\max_{j\in[M]_0}\rho_{m,j}$, which is why $t_0$ is dealt with separately from other template parameters in the conventional GW data analysis.
On the other hand, for the other intrinsic parameters, we calculate $\rho_m$ for each $m\in[N_{\rm temp}]_0$ one by one until we get $\rho_m\ge\rho_{\rm th}$, which results in the number of floating-point operations of order
\begin{equation}
O\left(\frac{M\log M}{r(\rho_{\rm th})}\right) \label{eq:compClass1}
\end{equation}
for $r(\rho_{\rm th})>0$. Here
\begin{equation}
	r(\rho):=\frac{\left|\left\{m\in[N_{\rm temp}]_0 \ \middle| \ \rho_{m}\ge\rho \right\}\right|}{N_{\rm temp}} \label{eq:r}
\end{equation}
is the fraction of intrinsic parameter sets in the template banks that yields SNRs larger than $\rho$, with time of coalescence optimized.
When $r(\rho_{\rm th})=0$, we have to go through all the template and thus the number of floating-point operations is
\begin{equation}
O\left(N_{\rm temp}M\log M\right) \,. 
\label{eq:compClass2}
\end{equation}

Unlike the above, as we will see later, we do not use FFT in the proposed quantum method.
For later convenience, we rewrite Eq.~(\ref{eq:SNRdiscorg}) as
\begin{eqnarray}
	\rho_{m,j}&=&\frac{2}{M}\sum^{\frac{M}{2}-1}_{k=1} \tilde{\rho}_{m,j,k}, \nonumber\\
	\tilde{\rho}_{m,j,k}&:=&\Re \left(\frac{2\tilde{Q}^*_m(f_k)\tilde{s}(f_k)}{S_{\rm n}(f_k)\Delta t}\right)\cos\left(\frac{2\pi jk}{M}\right)-\Im \left(\frac{2\tilde{Q}^*_m(f_k)\tilde{s}(f_k)}{S_{\rm n}(f_k)\Delta t}\right)\sin\left(\frac{2\pi jk}{M}\right),
	\label{eq:SNRdisc}
\end{eqnarray}
and set $\tilde{\rho}_{m,j,0}=0$.

\subsection{Representation of real numbers on qubits and some basic quantum circuits}

In numerical calculations in this paper, we use bit strings on qubits as fixed-point binary representations of real numbers and, for $x\in\mathbb{R}$, we denote by $\ket{x}$ the computational basis state on a quantum register which corresponds to $x$.
Unless otherwise stated, we assume that each register has $N_{\rm dig}$ qubits, where $N_{\rm dig}$ is a sufficiently large positive integer set independently from parameters in problems under consideration, and neglect errors from finite-precision representation.

For computing with real numbers, we use the quantum circuits for four basic arithmetic operations: addition $O_{\rm add}\ket{x}\ket{y}\ket{0}=\ket{x}\ket{y}\ket{x+y}$, subtraction $O_{\rm sub}\ket{x}\ket{y}\ket{0}=\ket{x}\ket{y}\ket{x-y}$, multiplication $O_{\rm mul}\ket{x}\ket{y}\ket{0}=\ket{x}\ket{y}\ket{xy}$ and division $O_{\rm div}\ket{x}\ket{y}\ket{0}=\ket{x}\ket{y}\ket{x/y}$, where $x,y$ are any real numbers ($y\ne 0$ for division) and some ancillary registers may be undisplayed.
In fact, concrete circuit implementations for these operation have been presented \cite{vedral1996,beckman1996,draper2000,cuccaro2004,takahashi2005,draper2006,alvarez2008,takahashi2008,takahashi2010,khosropour2011,jamal2013,jayashree2016,dibbo2016,babu2017,munoz2018,thapliyal2019}.
These circuits enable the calculation of rational functions.
We also use quantum circuits for calculation of elementary functions $f$ such as exponential, trigonometric functions, and so on: $O_f\ket{x}\ket{0}=\ket{x}\ket{f(x)}$ for any $x\in\mathbb{R}$.
Such circuits can be implemented through, for example, piecewise polynomial approximation \cite{haner2018}.

In addition to these circuits for numerical calculation, we now list some oracles used in the proposed quantum algorithm.
A comparer $O_{\rm comp}$ acts as $O_{\rm comp}\ket{x}\ket{y}\ket{0}=\ket{x}\ket{y}\left(\mathbbm{1}_{x\ge y}\ket{1}+\mathbbm{1}_{x<y}\ket{0}\right)$ for any $x,y\in\mathbb{R}$.
This is actually equivalent to subtraction $x-y$, since, if we adopt 2's complement method to represent negative numbers, the most significant digit represents the sign of a number \cite{koren2001}.
A Y-rotation with controlled angle gate $O_{\rm CY}$ acts as $O_{\rm CY}\ket{\theta}\ket{\psi}=\ket{\theta}\otimes R_{\rm Y}(\theta)\ket{\psi}$, where $R_{\rm Y}(\theta):=\begin{pmatrix}\cos\frac{\theta}{2} & -\sin\frac{\theta}{2} \\ \sin\frac{\theta}{2} & \cos\frac{\theta}{2} \end{pmatrix}$, for any $\theta\in\mathbb{R}$ and any single-qubit state $\ket{\psi}$.
This is implemented as a sequence of fixed-angle controlled Y-rotation gates $\ket{0}\bra{0}\otimes \mathbb{I}_2+\ket{1}\bra{1}\otimes R_{\rm Y}(\theta)$\cite{egger2020}, where $\theta\in\mathbb{R}$ is prefixed.
We also use the oracle $O^{\rm EqPr}_N$, where $N\in\mathbb{N}$, to generate equiprobable superposition of states from $\ket{0}$ to $\ket{N-1}$: $O^{\rm EqPr}_N\ket{0}=\frac{1}{\sqrt{N}}\sum_{i=0}^{N-1} \ket{i}$.
If $N=2^n$ with some $n\in\mathbb{N}$, we can generate such a state by operating a Hadamard gate on each qubit in a $n$-qubit register.
Also for $N$ that is not a power of 2, we can implement $O^{\rm EqPr}_N$ by the method in \cite{grover2002} to generate a state in which a given probability density $p(x)$ is amplitude-encoded, with $p(x)$ set to the uniform density on $[0, (N-1)/2^n]$ with $n=\left\lceil\log_2 N \right\rceil$.

The last oracle we mention here is $O^{\rm med}_N$ that, for any $N$ real numbers $x_1,...,x_N$, outputs the median ${\rm med}(x_1,...,x_N)$: $O^{\rm med}_N\ket{x_1}\cdots\ket{x_N}\ket{0}=\ket{x_1}\cdots\ket{x_N}\ket{{\rm med}(x_1,...,x_N)}$. 
This operation is implemented as follows.
First, we transform $\ket{x_1}\cdots\ket{x_N}$ to $\ket{x^{\rm sort}_1}\cdots\ket{x^{\rm sort}_N}$, where $x^{\rm sort}_1,...,x^{\rm sort}_N$ is a sequence made by ascending sort of $x_1,...,x_N$ and some ancillary qubits are not displayed\footnote{In this ascending sort operation, unitarity holds in the system including ancillary qubits.}.
Then, we let the number on the midmost register be ${\rm med}(x_1,...,x_N)$.
Note that exchange based sort algorithms such as bubble sort \cite{numericalrecipe} can be implemented since the operation
\begin{equation}
    \ket{x}\ket{y}\ket{0}\rightarrow \ket{x}\ket{y}\ket{\mathbbm{1}_{x\ge y}} \rightarrow \begin{cases}
    \ket{y}\ket{x}\ket{1} & ; \ {\rm if} \  x\ge y \\
    \ket{x}\ket{y}\ket{0} & ; \ {\rm otherwise}
    \end{cases},
\end{equation}
is possible for any $x,y\in\mathbb{R}$, where we use $O_{\rm comp}$ and controlled SWAP gates at the first and second arrows, respectively.

Hereafter, we collectively call the above oracles arithmetic oracles.

\subsection{Grover's algorithm and Quantum amplitude amplification}

Grover's algorithm \cite{grover1996} is a quantum algorithm for searching a ``marked entry" $x$, which satisfies some condition given as a binary-valued function, from an unstructured database.
Formally, we have the following theorem.

\begin{theorem}
Let $n\in\mathbb{N}$ and $F:\{0,1\}^{n}\rightarrow\{0,1\}$ is a function such that $F(x_{\rm tar})=1$ for one element $x_{\rm tar}\in\{0,1\}^{n}$ and $F(x)=0$ for any $x\in\{0,1\}^{n}\setminus\{x_{\rm tar}\}$.
Suppose that we are given an access to an oracle $O_F$ on a system consisting of a $n$-qubit register and a qubit such that $O_F\ket{x}\ket{0}=O_F\ket{x}\ket{F(x)}$ for any $x\in\{0,1\}^{n}$.
Then, for any $\delta\in(0,1)$, there exists a quantum algorithm that, with probability at least $1-\delta$, outputs $x_{\rm tar}$ making $O\left(\sqrt{N}\log\delta^{-1}\right)$ calls to $O_F$, where $N:=2^n$.
\label{th:Grover}
\end{theorem}

This is often called a quadratic speedup over classical search methods that takes $O(N)$ time for the same problem.

Besides, there exists a quantum algorithm called QAA \cite{brassard1998,brassard2002}, which can be seen as an extension of Grover's algorithm.
It is an algorithm to amplify the amplitude of the ``marked state" in a given superposition and obtain the state.
Formally, the following theorem holds.

\begin{theorem}
    Suppose that we are given an access to an oracle $A$ that acts on a system consisting of a $n$-qubit register and a single-qubit register as
    \begin{equation}
        A\ket{0}\ket{0}=\sqrt{a}\ket{\phi_1}\ket{1}+\sqrt{1-a}\ket{\phi_0}\ket{0}=:\ket{\Phi}, \label{eq:A}
    \end{equation}
    where $\ket{\phi_0}$ and $\ket{\phi_1}$ are some quantum states on the register and $a\in[0,1)$.
    Then, for any $\gamma,\delta\in(0,1)$, there exists a quantum algorithm $\proc{QAA}(A,\gamma,\delta)$ that uses $O(n)$ qubits and behaves as follows:
    \begin{itemize}
        \item The output of the algorithm is either of
        \begin{enumerate}
            \renewcommand{\labelenumi}{(\Alph{enumi})}
            \item the message ``success" and the quantum state $\ket{\phi_1}$
            \item the message ``failure"
        \end{enumerate}
        \item If $a\ge\gamma$, the algorithm outputs (A) with probability at least $1-\delta$, making $O\left(\frac{\log\delta^{-1}}{\sqrt{a}}\right)$ queries to $A$.
        \item If $a<\gamma$, the algorithm outputs either (A) or (B), making $O\left(\frac{\log\delta^{-1}}{\sqrt{\gamma}}\right)$ queries to $A$.
    \end{itemize}
    \label{th:QAA}
\end{theorem}

\begin{algorithm}[htp]
	\caption{$\proc{QAA}(A,\gamma,\delta)$, a modified version of \textbf{QSearch} in \cite{brassard2002} with $c=3/2$} 
	\label{alg:QAA}
	\begin{algorithmic}[1]
		\REQUIRE{$A$ in Eq.~(\ref{eq:A}), $G$ in Eq.~(\ref{eq:G}), $\gamma\in(0,1)$, $\delta\in(0,1)$}
		
		\STATE Set $L=\left\lceil\log_{\frac{3}{2}} \frac{3}{4\sqrt{\gamma}}\right\rceil$, $m_1=\left\lceil\log_{\frac{1}{4}} \delta\right\rceil$ and $m_2=\left\lceil\log_{\frac{5}{6}} \delta\right\rceil$.
		
		\FOR{$i=1$ to $m_1$} 

		    \STATE Generate $\ket{\Phi}$ and measure the single-qubit register. If the outcome is 1, output (A). 
		
		\ENDFOR
		
		\FOR{$l=1$ to $L$}
		
		\STATE Randomly choose an integer $j\in [M]$, where $M=\left\lceil\left(\frac{3}{2}\right)^l\right\rceil$.
		
		\FOR{$i=1$ to $m_2$}
		\STATE Generate $G^j\ket{\Phi}$ and measure the single-qubit register. If the outcome is 1, output (A).
		\ENDFOR
		    
        \ENDFOR
        
        \STATE Output (B).

	\end{algorithmic}
\end{algorithm}

The procedure of $\proc{QAA}(A,\gamma,\delta)$ is presented in Algorithm \ref{alg:QAA}.
Here, $G$, the so-called Grover operator, is defined as
\begin{equation}
    G:=-AS_0A^{-1}S_\chi.
\end{equation}
$S_\chi$ is an operator that acts as $S_\chi\ket{\psi}\ket{0}=\ket{\psi}\ket{0}$ and $S_\chi\ket{\psi}\ket{1}=-\ket{\psi}\ket{0}$, where $\ket{\psi}$ is any state on the $n$-qubit register, and implemented just as a Z gate on the single-qubit register.
The operator $S_0$ acts as $S_0\ket{0}\ket{0}=-\ket{0}\ket{0}$ and $S_0\ket{\Psi}=\ket{\Psi}$ for any other computational basis states $\ket{\Psi}$ on the system.
We can implement this using a multi-controlled Z gate.

Before presenting the proof of Theorem \ref{th:QAA}, let us roughly see how Algorithm \ref{alg:QAA} works.
We can show that, for any $j\in\mathbb{N}$,
\begin{equation}
    G^j\ket{\Phi}=\sin((2j+1)\theta_a)\ket{\phi_1}\ket{1}+\cos((2j+1)\theta_a)\ket{\phi_0}\ket{0} \label{eq:G}
\end{equation}
holds, where $\theta_a=\arcsin(\sqrt{a})$ \cite{brassard2002}.
Therefore, operating $G$ $O(1/\sqrt{a})$ times on $\ket{\Phi}$ makes the amplitude of $\ket{\phi_1}\ket{1}$ of order 1, which means high probability to obtain 1 on the single-qubit register.

Now, the proof of Theorem \ref{th:QAA} is as follows.

\begin{proof}[Proof of Theorem \ref{th:QAA}]
To begin with, note some differences between Algorithm \ref{alg:QAA} and \textbf{QSearch} in \cite{brassard2002}.
First, in Algorithm \ref{alg:QAA}, loop 5-10 \footnote{Here, loop $a$-$b$ means that the loop from line $a$ to line $b$ in Algorithm \ref{alg:QAA}.} has a bound $L$ on the iteration number, whereas \textbf{QSearch} has no bound in the corresponding loop and can run forever.
Second, Algorithm \ref{alg:QAA} repeats state generations and measurements in loop 2-4 and loop 7-9, whereas in \textbf{QSearch} they are not repeated.

Under these differences, Algorithm \ref{alg:QAA} behaves as follows.
If $a\ge \frac{3}{4}$, loop 2-4 outputs (A) with probability at least
\begin{equation}
1-\left(1-a\right)^{m_1}\ge 1-\left(\frac{1}{4}\right)^{m_1}\ge 1-\left(\frac{1}{4}\right)^{\log_{\frac{1}{4}} \delta}\ge 1-\delta.
\end{equation}
In this, $A$ is called at most $m_1=O(\log \delta^{-1})$ times, regardless of the value of $\gamma$.

On the other hand, if $\gamma\le a< \frac{3}{4}$, the algorithm works as follows.
Loop 7-9 with $l=\tilde{l}(a):=\left\lceil\log_{\frac{3}{2}} \frac{3}{4\sqrt{a}}\right\rceil$ outputs (A) with probability at least
\begin{equation}
    1-\left[1-\frac{1}{2}\left(1-\frac{1}{2M\sqrt{a}}\right)\right]^{m_2}=1-\left(\frac{1}{2}+\frac{1}{4M\sqrt{a}}\right)^{m_2}\ge 1-\left(\frac{1}{2}+\frac{1}{4\sqrt{a}\left(\frac{3}{2}\right)^{\tilde{l}(a)}}\right)^{m_2}\ge 1-\left(\frac{1}{2}+\frac{1}{3}\right)^{\log_{\frac{5}{6}}\delta}=1-\delta,
\end{equation}
since, according to \cite{brassard2002}, one run of line 8 outputs (A) with probability at least $\frac{1}{2}\left(1-\frac{1}{2M\sqrt{a}}\right)$ if $0<a<3/4$.
The number of queries to $A$ until we get (A) is evaluated as follows.
Since $G$ contains two calls to A, loop 7-9 with $l=l^\prime$ makes at most $O\left(m_2\left(\frac{3}{2}\right)^{l^\prime}\right)$ queries to $A$ for generation of $G^j\ket{\Phi}$.
Therefore, until we get (A), $A$ is called $O\left(\sum_{l=1}^{\tilde{l}(a)}\left(\frac{3}{2}\right)^lm_2\right)$ times, that is, $O\left(\log\delta^{-1}/\sqrt{a}\right)$ times.

In summary, if $\gamma\le a\le 1$, Algorithm \ref{alg:QAA} outputs (A) with probability at least $1-\delta$ making $O\left(\log\delta^{-1}/\sqrt{a}\right)$ queries to $A$.

To show the statement on the case that $a<\gamma$, we need only to show that the maximum number of queries to $A$ in this algorithm is $O(\log \delta^{-1}/\sqrt{\gamma})$.
This is actually true, since the number of queries to $A$ in loop 5-10 is $O\left(\sum_{l=1}^{L}\left(\frac{3}{2}\right)^lm_2\right)$, that is, $O(\log \delta^{-1}/\sqrt{\gamma})$, and adding $O(\log \delta^{-1})$ queries in loop 2-4 does not change the order.

The statement on qubit number is obvious, since every operation in Algorithm \ref{alg:QAA} is done by $A$ or $G$, which is an operator on the system consisting of a $n$-qubit register and a single qubit register.
\end{proof}

Let us make some comments on QAA.
First, note that QAA can be in fact regarded as an extension of Grover's algorithm, since the search problem in Theorem \ref{th:Grover} can be solved by QAA.
This is because we can generate the following state by $O^{\rm EqPr}_N$ and $O_F$
\begin{equation}
    \frac{1}{\sqrt{N}}\sum_{x\in\{0,1\}^n}\ket{x}\ket{F(x)}=\frac{1}{\sqrt{N}}\ket{x_{\rm tar}}\ket{1}+\frac{1}{\sqrt{N}}\sum_{x\in\{0,1\}^n\setminus\{x_{\rm tar}\}}\ket{x}\ket{0},
\end{equation}
which is in the form of Eq.~(\ref{eq:A}) with $\ket{\phi_1}=\ket{x_{\rm tar}}$.
Second, QAA provides a quadratic speedup like Grover's algorithm.
Instead of QAA, we can repeat generating $\ket{\Phi}$ and measuring the qubit until we get the measurement outcome 1 and the state $\ket{\phi_1}$.
This naive way yields $O(a^{-1})$ repetitions in expectation.
Therefore, QAA is quadratically faster than this.
Third, note that, in Theorem \ref{th:QAA}, the marked state is defined as a state in which some qubit takes $\ket{1}$.
Although in the original algorithm the marked state can be set more generally \cite{brassard2002}, the above setting is sufficient for the proposed algorithm for GW matched filtering, as we will see later.

\subsection{Quantum amplitude estimation}

Based on QAA, we can construct an algorithm called QAE for estimating the amplitude of a target state in a superposition state, or, more specifically, $a$ in the state like Eq.~(\ref{eq:A}) \cite{brassard2002}.
Roughly speaking, in the algorithm, we generate a superposition of states in the form of $G^j\ket{\Phi}$ with various values of $j$ by iteratively operating $G$ controlled by some register $R_{\rm QFT}$, and outputs an approximation of $\theta_a$ onto $R_{\rm QFT}$ by quantum Fourier transform (QFT)\footnote{There are some variants of QAE that rely on not QFT but iterative measurements and processing outcomes \cite{suzuki2020,aaronson2020,nakaji2020,grinko2021,tanaka2021,uno2021,giurgica2020,wang2021,tanaka2022,giurgica2021}. However, we do not use these in this paper since, in the proposed algorithm, we require QAE to be a unitary operation as a subroutine in QAA, as mentioned below.}.
Here, we do not enter details of the procedure but just present the following theorem, which is a modification of Theorem 12 in \cite{brassard2002}, without proof.

\begin{theorem}
    Suppose that we are given an access to an oracle $A$ in Eq.~(\ref{eq:A}).
    Then, for any integer $t$ larger than 2, there is an oracle $\tilde{O}_{A,t}^{\rm QAE}$ that acts as $\tilde{O}_{A,t}^{\rm QAE}\ket{0}=\sum_{y\in \mathcal{Y}} \alpha_y\ket{y}$, where some ancillary qubits are undisplayed.
    Here, $\mathcal{Y}$ is a finite set of real numbers that includes a subset $\tilde{\mathcal{Y}}$ consisting of elements $\tilde{a}$ satisfying
    \begin{equation}
        |\tilde{a}-a|\le \frac{2\pi\sqrt{a(1-a)}}{t}+\frac{\pi^2}{t^2}, \label{eq:QAEErr}
    \end{equation}
    and $\{\alpha_y\}_{y\in\mathcal{Y}}$ are complex numbers satisfying $\sum_{\tilde{y}\in\tilde{\mathcal{Y}}}|\alpha_{\tilde{y}}|^2\ge 8/\pi^2$.
    In $\tilde{O}_{A,t}^{\rm QAE}$, $O_{\mathcal{X}}$ is used $O\left(t\right)$ times and $O(n+\log t)$ qubits are used.
    \label{th:QAE}
\end{theorem}

Here are some comments.
In the original algorithm \textbf{Est\_Amp} in \cite{brassard2002}, we measure the state $\sum_{y\in \mathcal{Y}} \alpha_y\ket{y}$ and obtain an estimate on $a$.
However, we now stop the procedure at generation of the state, since, as explained below, we use QAE as a subroutine for the SNR calculation in searching high SNR templates by QAA and thus require it to be a unitary operation.
Besides, note that the statement on qubit number is obvious since \textbf{Est\_Amp} in \cite{brassard2002} uses only the register $R_{\rm QFT}$, which has $O(\log t)$ qubits, along with the system on which $A$ acts.

We often want to enhance the lower bound $8/\pi^2$ on the success probability of QAE to a given high value.
We can accomplish this thanks to the following theorem, which is Lemma 1 in \cite{montanaro2015} and originally Lemma 6.1 in \cite{jerrum1986}.
\begin{theorem}
    Let $\mu\in\mathbb{R}$ and $\epsilon\in\mathbb{R}_+$.
    Let $\mathcal{A}$ be an algorithm that outputs an $\epsilon$-approximation of $\mu$ with probability $\gamma\ge \frac{3}{4}$.
    Then, for any $\delta\in(0,1)$, the median of outputs in $12\left\lceil\log\delta^{-1}\right\rceil+1$ runs of $\mathcal{A}$ is an $\epsilon$-approximation of $\mu$ with probability at least $1-\delta$.
    \label{th:median}
\end{theorem}

This implies the following.
Letting $N$ be an integer larger than $12\left\lceil\log\delta^{-1}\right\rceil+1$, we generate the state
\begin{equation}
    \ket{\Psi_N}:=\sum_{y_1,...,y_N\in \mathcal{Y}} \alpha_{y_1}\cdots\alpha_{y_N}\ket{y_1}\cdots\ket{y_N}\ket{{\rm med}(y_1,...,y_N)}
\end{equation}
by operating $\tilde{O}_{A,t}^{\rm QAE}$ on each of the first $N$ registers and using $O^{\rm med}_N$.
We rewrite this state as $\ket{\Psi_N}:= \alpha_{\le\epsilon}\ket{\psi_{\le\epsilon}}+\alpha_{>\epsilon}\ket{\psi_{>\epsilon}}$ with $\alpha_{\le\epsilon},\alpha_{>\epsilon}\in\mathbb{C}$ and the states
\begin{equation}
    \ket{\psi_{\le\epsilon}}:= \sum_{\substack{\tilde{a}\in \mathcal{Z} \\ |\tilde{a}-a|\le \frac{2\pi\sqrt{a(1-a)}}{t}+\frac{\pi^2}{t^2}}} \beta_{\tilde{a}}\ket{\phi_{\tilde{a}}}\ket{\tilde{a}}\,,\qquad \ket{\psi_{>\epsilon}}:= \sum_{\substack{\tilde{a}\in \mathcal{Z} \\ |\tilde{a}-a|> \frac{2\pi\sqrt{a(1-a)}}{t}+\frac{\pi^2}{t^2}}} \beta_{\tilde{a}}\ket{\phi_{\tilde{a}}}\ket{\tilde{a}},
\end{equation}
where $\mathcal{Z}$ is some finite set of real numbers, $\{\beta_{\tilde{a}}\}_{\tilde{a}\in\mathcal{Z}}$ are complex numbers, and $\{\ket{\phi_{\tilde{a}}}\}_{\tilde{a}\in\mathcal{Z}}$ are states on the first $N$ registers.
Then, $|\alpha_{\le\epsilon}|^2$, the squared amplitude of the state $\ket{\psi_{\le\epsilon}}$, in which the number $\tilde{a}$ on the last register satisfies Eq.~(\ref{eq:QAEErr}), is larger than $1-\delta$.
This technique is used in the quantum algorithm for Monte Carlo integration, which is explained next.

\subsection{Quantum Monte Carlo integration}

On the basis of QAE, we can construct a quantum algorithm for Monte Carlo integration \cite{montanaro2015}, which we hereafter call QMCI.
Although Monte Carlo integration is generally a method to estimate integrals, we now consider it as a method to estimate the mean of $\mathcal{X}$ a given set of real numbers, since it is sufficient for the proposed algorithm.
Among some versions presented in \cite{montanaro2015}, we use the one for the situation where an upper bound on ${\rm Var}(\mathcal{X})$ is given.

\begin{theorem}
    Let $N\in\mathbb{N}$ and $\mathcal{X}$ be a set of $N$ real numbers, $X_0,...,X_{N-1}$, whose mean is $\mu:=\frac{1}{N}\sum_{i=0}^{N-1}X_i$ and sample variance ${\rm Var}(\mathcal{X})$ satisfies ${\rm Var}(\mathcal{X})\le\sigma^2$ with $\sigma\in\mathbb{R}_+$.
    Suppose that we are given an oracle $O_X$ that acts on a system with $O(\log N)$ qubits in total as
    \begin{equation}
        O_{\mathcal{X}}\ket{i}\ket{0}=\ket{i}\ket{X_i}, \label{eq:OX}
    \end{equation}
    for any $i\in[N]_0$.
    Let $\epsilon\in(0,4\sigma)$ and $\delta\in(0,1)$.
    Then, there is an oracle $O_{\mathcal{X},\epsilon,\delta,\sigma}^{\rm mean}$ such that
    \begin{equation}
        O_{\mathcal{X},\epsilon,\delta,\sigma}^{\rm mean}\ket{0}=\sum_{y\in \mathcal{Y}} \alpha_y\ket{y}, \label{eq:ObarX}
    \end{equation}
    where some ancillary qubits are undisplayed.
    Here, $\mathcal{Y}$ is a finite set of real numbers that includes a subset $\tilde{\mathcal{Y}}$ consisting of $\epsilon$-approximations of $\mu$ and $\{\alpha_y\}_{y\in\mathcal{Y}}$ are complex numbers satisfying $\sum_{\tilde{y}\in\tilde{\mathcal{Y}}}|\alpha_{\tilde{y}}|^2\ge 1-\delta$.
    In $O_{\mathcal{X},\epsilon,\delta,\sigma}^{\rm mean}$,
    \begin{equation}
        O\left(\frac{\sigma}{\epsilon}\log^{3/2}\left(\frac{\sigma}{\epsilon}\right)\log\log\left(\frac{\sigma}{\epsilon}\right)\log\left(\frac{1}{\delta}\right)\right) \label{eq:compQMCI}
    \end{equation}
    queries to $O_{\mathcal{X}}$ are made and
    \begin{equation}
        O\left(\left(\log N+\log\left(\frac{\sigma}{\epsilon}\right)\right)\log\left(\frac{\sigma}{\epsilon}\right)\log\log\left(\frac{\sigma}{\epsilon}\right)\log\delta^{-1}\right) \label{eq:qubitQMCI}
    \end{equation}
    qubits are used.
    \label{th:QMCI}
\end{theorem}
We present the proof in \ref{sec:proofQMCI}.

\section{Quantum algorithm for gravitational wave matched filtering \label{sec:main}}

\subsection{Problem and assumptions}

Equipped with the above quantum algorithms, we now consider applying them to GW matched filtering.
We start from formally stating the problem we consider. 

\begin{problem}
Let $T,M$ and $N_{\rm temp}$ be a positive real number, a positive even integer and a positive integer, respectively.
Define $\Delta t:=T/M$ and, for $k\in[M/2-1]$, $f_k:=k/T$.
Suppose that we are given a complex sequence $\{\tilde{s}(f_k)\}_{k\in[M/2-1]}$, a function $S_n:\mathbb{R}_+\rightarrow\mathbb{R}_+$, and, for every $m\in[N_{\rm temp}]_0$, a function $\tilde{Q}_m:\mathbb{R}_+\rightarrow\mathbb{C}$.
Then, determine whether there exists any $(m,j)\in[N_{\rm temp}]_0\times\left[M\right]_0$ such that $\rho_{m,j}$ in Eq.~(\ref{eq:SNRdisc}) exceeds some given value or not.
If there are such integer pairs, find one of them.
\label{prob:GWMF}
\end{problem}

We need some preparations to tackle this problem.
First, let us define some quantities for convenience.
The first one is as follows: for $\rho\in\mathbb{R}_+$,
\begin{equation}
	\tilde{r}(\rho):=\frac{\#\left\{(m,j)\in[N_{\rm temp}]_0\times[M]_0 \ \middle| \ \rho_{m,j}\ge\rho \right\}}{N_{\rm temp}M}.
\end{equation}
Obviously, $\tilde{r}(\rho)$ represents the fraction of templates that give SNRs larger than $\rho$ and is an analog of $r(\rho)$ in Eq.~(\ref{eq:r}).
The next one is about the magnitude of template functions $\tilde{Q}_m$ normalized by $S_{\rm n}$: 
\begin{equation}
    \gamma:=\max_{(m,k)\in[N_{\rm temp}]_0\times\left[\frac{M}{2}-1\right]}\frac{|\tilde{Q}_m(f_k)|}{\sqrt{S_{\rm n}(f_k)\Delta t}}.
    \label{eq:gamma_def}
\end{equation}
As we will discuss in Section \ref{sec:gamma}, we expect that this is of order $O(1)$.

Next, let us make some assumptions needed to discuss the quantum algorithm and its complexity.
The first one is about availability of some fundamental oracles.
\begin{assum}	
	We have accesses to oracles $O_{\rm Re}$ and $O_{\rm Im}$ such that, for every $(m,k)\in[N_{\rm temp}]_0\times\left[\frac{M}{2}\right]_0$,
 	\begin{eqnarray}
		O_{\rm Re}\ket{m}\ket{k}\ket{0}&=&
		\begin{cases}
		\ket{m}\ket{k}\ket{0} & ; \ {\rm if} \ k=0 \\
		\ket{m}\ket{k}\Ket{\Re\left(\frac{2\tilde{Q}^*_m(f_k)\tilde{s}(f_k)}{S_{\rm n}(f_k)\Delta t}\right)} &;  \ {\rm otherwise}
		\end{cases},
		\nonumber \\
		O_{\rm Im}\ket{m}\ket{k}\ket{0}&=&
		\begin{cases}
		\ket{m}\ket{k}\ket{0} & ; \ {\rm if} \ k=0 \\
		\ket{m}\ket{k}\Ket{\Im\left(\frac{2\tilde{Q}^*_m(f_k)\tilde{s}(f_k)}{S_{\rm n}(f_k)\Delta t}\right)} & ; \ {\rm otherwise}
		\end{cases}
	\end{eqnarray}
	\label{ass:oracle}
\end{assum}
We will discuss how to implement these in Section \ref{sec:oracle}.
The next assumption is on the mean and the variance of detector outputs used in the SNR calculation.
\begin{assum}
    \begin{eqnarray}
        &&\left({\rm Mean}\left(\left\{\frac{2\Re\tilde{s}(f_k)}{\sqrt{S_{\rm n}(f_k)T}}\right\}_{k=1,...,M/2-1}\right)\right)^2\le1, \nonumber \\
        &&\left({\rm Mean}\left(\left\{\frac{2\Im\tilde{s}(f_k)}{\sqrt{S_{\rm n}(f_k)T}}\right\}_{k=1,...,M/2-1}\right)\right)^2\le1, \nonumber \\
        &&{\rm Var}\left(\left\{\frac{2\Re\tilde{s}(f_k)}{\sqrt{S_{\rm n}(f_k)T}}\right\}_{k=1,...,M/2-1}\right)\le4, \nonumber \\
        &&{\rm Var}\left(\left\{\frac{2\Im\tilde{s}(f_k)}{\sqrt{S_{\rm n}(f_k)T}}\right\}_{k=1,...,M/2-1}\right)\le4.
    \end{eqnarray}
    \label{ass:hVar}
\end{assum}
We will discuss the validity of this in Section \ref{sec:hvar}.
As we will see in Section \ref{sec:propAlg}, this assumption is important for the proposed quantum algorithm, since it leads to the upper bound on the variance of summands in the SNR calculation, with which we can use QMCI for variables with bounded variance.

\subsection{Previous algorithm} \label{sec:previous}

Before the new quantum algorithm, we review the algorithm proposed in \cite{gao2021}.
It is shown as Algorithm \ref{alg:prev}, which is a modified version of Algorithm 2 in \cite{gao2021}.
Given a SNR threshold $\rho_{\rm th}$, this algorithm outputs the message "there is a signal" and $m\in[N_{\rm temp}]_0$ such that $\rho_m\ge\rho_{\rm th}$ with probability at least $1-\delta$, if there exists such $m$.
Note that Algorithm \ref{alg:prev} in our paper uses QAA instead of QAE on $O^\prime_{\rm FFT}\ket{0}\ket{0}\ket{0}$ in Algorithm 2 of \cite{gao2021}, which does not affect the scaling of the complexity on $M$ and $N_{\rm temp}$.
Also note that we can implement $O_{\rm FFT}$ in Eq.~(\ref{eq:OFFT}) by arithmetic oracles, since FFT is actually a sequence of arithmetic operations\footnote{In fact, implementation of FFT as a quantum circuit has been studied in \cite{asaka2020}.}. 

\begin{algorithm}[htp]
	\caption{Previous algorithm for GW matched filtering (modified)} 
	\label{alg:prev}
	\begin{algorithmic}[1]
		\REQUIRE{\ \\
		$\delta\in(0,1)$.\\
		$\rho_{\rm th}\in\mathbb{R}_+$ the SNR threshold. \ \\
		An oracle $O_{\rm FFT}$ such that, for every $m\in[N_{\rm temp}]_0$,
		\begin{equation}
		O_{\rm FFT}\ket{m}\ket{0}=\ket{m}\ket{\rho_m}. \label{eq:OFFT}
		\end{equation}
		}
		
		\STATE Combining $O^{\rm EqPr}_{N_{\rm temp}}$, $O_{\rm FFT}$ and a comparer, construct an oracle $O^\prime_{\rm FFT}$ that acts as
		\begin{equation}
		    O^\prime_{\rm FFT}\ket{0}\ket{0}\ket{0}=\frac{1}{\sqrt{N_{\rm temp}}}\sum_{m=0}^{N_{\rm temp}-1}\ket{m}\ket{\rho_m}\ket{\mathbbm{1}_{\rho_m\ge\rho_{\rm th}}}.
		\end{equation}
		
		\STATE Run $\proc{QAA}\left(O^\prime_{\rm FFT},\frac{1}{N_{\rm temp}},\delta\right)$.
        
        \IF{we get the message ``failure"}
            \STATE Output the message "there is no signal".
        \ELSE 
            \STATE Measure the first register in the quantum state output by QAA and let the outcome be $m$.
            \STATE Calculate $\rho_{m}$ classically by FFT.
            \IF{$\rho_{m}\ge\rho_{\rm th}$}
                \STATE Output the message "there is a signal" and $m$.
            \ELSE
                \STATE Output the message "there is no signal".
            \ENDIF
        \ENDIF

	\end{algorithmic}
\end{algorithm}

The number of queries to arithmetic oracles in this algorithm is\footnote{Although Eqs.~(\ref{eq:compPrev1}) and (\ref{eq:compPrev2}) do not have a factor $\log N_{\rm temp}$, 
while Eq.~(35) in \cite{gao2021} has, we omit this reasonably assuming that $M\log M \gg\log N_{\rm temp}$.}
\begin{equation}
    O\left(\frac{M\log M}{\sqrt{r(\rho_{\rm th})}}\right), \label{eq:compPrev1}
\end{equation}
when $r(\rho_{\rm th})>0$ and
\begin{equation}
    O(\sqrt{N_{\rm temp}}M\log M) \label{eq:compPrev2}
\end{equation}
when $r(\rho_{\rm th})=0$.
We can see this as follows.
QAA makes $O(1/\sqrt{r(\rho_{\rm th})})$ calls to $O^\prime_{\rm FFT}$ and thus to $O_{\rm FFT}$ when $r(\rho_{\rm th})>0$, and $O(\sqrt{N_{\rm temp}})$ calls to them when $r(\rho_{\rm th})=0$.
Besides, the number of queries to arithmetic oracles in $O_{\rm FFT}$ is of order $O(M\log M)$ like the number of floating-point operations in FFT on a classical computer.
Combining these, we get the complexity bounds in Eqs.~(\ref{eq:compPrev1}) and (\ref{eq:compPrev2}).
They show a quadratic speedup over the classical complexity of  Eqs.~(\ref{eq:compClass1}) and (\ref{eq:compClass2}) with respect to $1/r(\rho_{\rm th})$ and $N_{\rm temp}$.

When it comes to qubit number, Algorithm 2 uses $O(M)$ qubits, since FFT calculates $\rho_{m,0},...,\rho_{m,M-1}$ simultaneously and thus use $O(M)$ registers to store intermediate and final calculation results.

\subsection{Proposed algorithm and its complexity \label{sec:propAlg}}

\subsubsection{Idea \label{sec:idea}}

The previous algorithm for GW matched filtering explained above uses FFT for the SNR calculation.
On the other hand, from the formula Eq.~(\ref{eq:SNRdisc}) for SNR, we conceive the following idea: can we use QMCI for the SNR calculation?
As we will see later, we can construct an oracle $O_\rho$ to calculate the summand $\tilde{\rho}_{m,j,k}$ in the SNR calculation in Eq.~(\ref{eq:SNRdisc}) making $O(1)$ uses of $O_{\rm Re}$ and $O_{\rm Im}$, and thus apply QMCI following Theorem \ref{th:QMCI}.
If we can set a bound $\sigma^2$ on the sample variance of $\{\tilde{\rho}_{m,j,k}\}_{k=1,...,\frac{M}{2}-1}$ and the accuracy $\epsilon$ in the SNR calculation, the query complexity with respect to $O_{\rm Re}$ and $O_{\rm Im}$ is of order $\widetilde{O}(\sigma/\epsilon)$.

This QMCI-based approach has a benefit on the qubit number reduction.
As shown in Eq.~(\ref{eq:qubitQMCI}), the number of qubits QMCI uses scales on $M$ as $O({\rm polylog} M)$, since the SNR given as Eq.~(\ref{eq:SNRdisc}) is a mean of $O(M)$ terms.
This means large reduction compared to the number required by FFT, which is of order $O(M)$.
This provides a large benefit, since quantum computers will have a limitation on qubit capacity even in the future as mentioned in Introduction.

When it comes to the query complexity, the QMCI-based method is roughly same as the previous FFT-based algorithm.
One might concern that the proposed method might worsen the scaling on $M$ due to the expansion of the space searched by QAA.
Unlike FFT, which simultaneously calculates $\rho_{m,j}$ for all $j\in[M]_0$, QMCI is performed for each $j$.
Thus, the search in the parameter space of $m\in[N_{\rm temp}]_0$ to find a large $\rho_m$ in the previous algorithm is replaced with the search in a larger parameter space of $(m,j)\in[N_{\rm temp}]_0\times[M]_0$ to find a large $\rho_{m,j}$.
This means that the iteration number in QAA increases from $\widetilde{O}(1/\sqrt{r(\rho_{\rm th})})$ to $\widetilde{O}(1/\sqrt{\tilde{r}(\rho_{\rm th})})$, which is $\widetilde{O}(\sigma\sqrt{M}/\epsilon\sqrt{r(\rho_{\rm th})})$ at maximum since $\tilde{r}(\rho_{\rm th}) > r(\rho_{\rm th})/M$.
Combining the aforementioned query complexity of the SNR calculation by QMCI, we bound the total number of queries to $O_{\rm Re}$ and $O_{\rm Im}$ as $\widetilde{O}\left(\sigma\sqrt{M}/\epsilon\sqrt{r(\rho_{\rm th})}\right)$.
Fortunately, since we can set $\sigma=O(\sqrt{M})$ and $\epsilon=O(1)$ as we will see below, the query complexity is resultingly $\widetilde{O}(M/\sqrt{r(\rho_{\rm th})})$, which is same as the number of queries to arithemetic oracles in the FFT-based method.
Note that this is a reasonable comparison, since $O_{\rm Re}$ and $O_{\rm Im}$ are constructed by arithmetic oracles, with the aid of the quantum random access memory (QRAM) \cite{giovannetti2008}, as explained later.

We should also note that the nature of QMCI causes the following issue.
The output of QMCI inevitably accompanies an error, and thus, even if the SNR of a given template calculated by QMCI exceeds the threshold $\rho_{\rm th}$, its true SNR might be below $\rho_{\rm th}$.
We may think that we can evade such a misjudge by setting the accuracy $\epsilon$ in QMCI extremely small, but it comes with large complexity.
Therefore, we need to reasonably set the accuracy: following the nature of the problem under consideration, we should derive the error tolerance for SNR and set the QMCI accuracy matching it.

For this, we propose the following way.
We consider GW matched filtering as a system that alarms us when the detector output seems to contain a signal.
Besides, we consider the two levels of SNR threshold denoted by $\rho_{\rm hard}$ and $\rho_{\rm soft}$, which have the following meaning.

\begin{itemize}
    \item If some templates have SNR $\rho\ge \rho_{\rm hard}$ for a given detector output, we want to be alarmed with certainty.
    \item We never want to be falsely alarmed when all templates have SNR $\rho<\rho_{\rm soft}$.
    \item 
    When no template has SNR $\rho\ge \rho_{\rm hard}$ but some have $\rho\in[\rho_{\rm soft},\rho_{\rm hard})$, it is not needed but fine to be alarmed.
\end{itemize}
\noindent In this situation, we set the QMCI accuracy to $(\rho_{\rm hard}-\rho_{\rm soft})/2$ and judge a template as matched if its SNR calculated by QMCI exceeds $(\rho_{\rm hard}+\rho_{\rm soft})/2$ and mismatched otherwise.
In this strategy, with high probability, a template with a true SNR $\rho\ge \rho_{\rm hard}$ is judged as matched and that with a true SNR $\rho< \rho_{\rm soft}$ is judged as mismatched.
A template with a true SNR $\rho_{\rm soft}\le\rho<\rho_{\rm hard}$ is an intermediate case where the data may contain a signal near the threshold and can be judged as either matched or mismatched due to the QMCI error.
We present an illustration of this strategy in Figure \ref{fig:twothres}.
We discuss how to set the two levels $\rho_{\rm hard}$ and $\rho_{\rm soft}$ in Section \ref{sec:twoThres}.

Note that "match" and "mismatch" discussed here are the result of quantum computation for a given detector output, and do not necessarily mean the result correctly indicates whether there is a GW signal or not.
In fact, a false detection can occur due to the random instrumental noise, which is not taken into account here.
It is, on the other hand, relevant to how we set the thresholds of  $\rho_{\rm hard}$ and $\rho_{\rm soft}$.
See Section \ref{sec:twoThres} for more detailed discussion. 

Based on this idea, we design a quantum algorithm for GW matched filtering in the following part.

\begin{figure}[tp]
		\centering
		\includegraphics[scale=0.8]{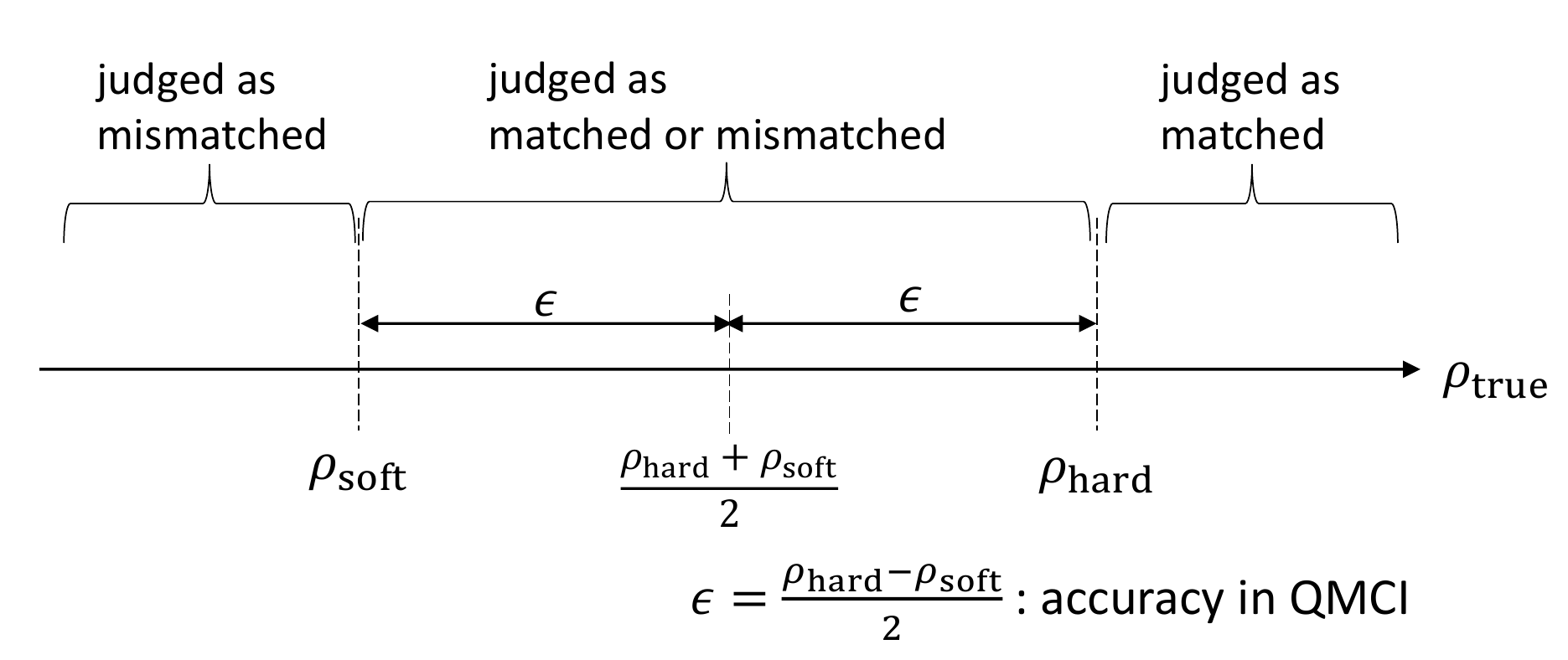}
		\caption{An illustration of the current strategy. For every template, if the SNR calculated by QMCI exceeds $(\rho_{\rm hard}+\rho_{\rm soft})/2$, we judge it as matched, and otherwise we judge it as mismatched. With QMCI accuracy $(\rho_{\rm hard}-\rho_{\rm soft})/2$, this leads to a correct judgement for templates with a true SNR $\rho_{\rm true}<\rho_{\rm soft}$ and $\rho_{\rm true}\ge\rho_{\rm hard}$ with high probability.}
	\label{fig:twothres}
\end{figure}

\subsubsection{Supporting lemma on the variance of summands in the SNR calculation}

Here, as a preparation to present the new algorithm, let us prove the following lemma on the variance of summands in the SNR calculation, which follows from Assumption \ref{ass:hVar}.

\begin{lemma}
    Under Assumption \ref{ass:hVar},
    \begin{equation}
        {\rm Var}\left(\{\tilde{\rho}_{m,j,k}\}_{k=0,...,M/2-1}\right)\le 80M\gamma^2 \label{eq:varrho}
    \end{equation}
    holds for every $(m,j)\in[N_{\rm temp}]_0\times\left[M\right]_0$.
    \label{lem:varrho}
\end{lemma}

\begin{proof}
    We can write $\tilde{\rho}_{m,j,k}=\tilde{\rho}^{(1)}_{m,j,k}+\tilde{\rho}^{(2)}_{m,j,k}+\tilde{\rho}^{(3)}_{m,j,k}+\tilde{\rho}^{(4)}_{m,j,k}$ with
    \begin{eqnarray}
        \tilde{\rho}^{(1)}_{m,j,k}&=&\sqrt{M}\Re\left(\frac{\tilde{Q}^*_m(f_k)}{\sqrt{S_{\rm n}(f_k)\Delta t}}\right)\frac{2\Re(\tilde{s}(f_k))}{\sqrt{S_{\rm n}(f_k)T}}\cos\left(\frac{2\pi jk}{M}\right), \nonumber \\
        \tilde{\rho}^{(2)}_{m,j,k}&=&-\sqrt{M}\Im\left(\frac{\tilde{Q}^*_m(f_k)}{\sqrt{S_{\rm n}(f_k)\Delta t}}\right)\frac{2\Im(\tilde{s}(f_k))}{\sqrt{S_{\rm n}(f_k)T}}\cos\left(\frac{2\pi jk}{M}\right), \nonumber \\
        \tilde{\rho}^{(3)}_{m,j,k}&=&-\sqrt{M}\Re \left(\frac{\tilde{Q}^*_m(f_k)}{\sqrt{S_{\rm n}(f_k)\Delta t}}\right)\frac{2\Im(\tilde{s}(f_k))}{\sqrt{S_{\rm n}(f_k)T}}\sin\left(\frac{2\pi jk}{M}\right), \nonumber \\
        \tilde{\rho}^{(4)}_{m,j,k}&=&-\sqrt{M}\Im \left(\frac{\tilde{Q}^*_m(f_k)}{\sqrt{S_{\rm n}(f_k)\Delta t}}\right)\frac{2\Re(\tilde{s}(f_k))}{\sqrt{S_{\rm n}(f_k)T}}\sin\left(\frac{2\pi jk}{M}\right).
    \end{eqnarray}
    We see that
    \begin{eqnarray}
        {\rm Var}\left(\left\{\tilde{\rho}^{(1)}_{m,j,k}\right\}_{k=0,...,M/2-1}\right)
        &\le&{\rm Mean}\left(\left\{\left(\tilde{\rho}^{(1)}_{m,j,k}\right)^2\right\}_{k=0,...,M/2-1}\right) \nonumber \\
        &\le&M\gamma^2{\rm Mean}\left(\left\{\left(\frac{2\Re(\tilde{s}(f_k))}{\sqrt{S_{\rm n}(f_k)T}}\right)^2\right\}_{k=1,...,M/2-1}\right) \nonumber \\
        &=&M\gamma^2\left[{\rm Var}\left(\left\{\frac{2\Re(\tilde{s}(f_k))}{\sqrt{S_{\rm n}(f_k)T}}\right\}_{k=1,...,M/2-1}\right)+\left({\rm Mean}\left(\left\{\frac{2\Re(\tilde{s}(f_k))}{\sqrt{S_{\rm n}(f_k)T}}\right\}_{k=1,...,M/2-1}\right)\right)^2\right] \nonumber \\
        &\le& 5M\gamma^2,
    \end{eqnarray}
    where the second inequality follows from
    \begin{equation}
        \left(\tilde{\rho}^{(1)}_{m,j,k}\right)^2= M \left(\Re\left(\frac{\tilde{Q}^*_m(f_k)}{\sqrt{S_{\rm n}(f_k)\Delta t}}\right)\right)^2\cos^2\left(\frac{2\pi jk}{M}\right)\left(\frac{2\Re(\tilde{s}(f_k))}{\sqrt{S_{\rm n}(f_k)T}}\right)^2\le M\gamma^2\left(\frac{2\Re(\tilde{s}(f_k))}{\sqrt{S_{\rm n}(f_k)T}}\right)^2 \,,
    \end{equation}
    for $k\in\left[\frac{M}{2}-1\right]$ and $\tilde{\rho}^{(1)}_{m,j,0}=0$, and the last inequality follows from Assumption \ref{ass:hVar}.
    Similarly, we have ${\rm Var}\left(\left\{\tilde{\rho}^{(2)}_{m,j,k}\right\}_{k=0,...,M/2-1}\right)\le5M\gamma^2$, ${\rm Var}\left(\left\{\tilde{\rho}^{(3)}_{m,j,k}\right\}_{k=0,...,M/2-1}\right)\le5M\gamma^2$ and ${\rm Var}\left(\left\{\tilde{\rho}^{(4)}_{m,j,k}\right\}_{k=0,...,M/2-1}\right)\le5M\gamma^2$.
    Combining these with
    \begin{eqnarray}
        &&{\rm Var}\left(\{\tilde{\rho}_{m,j,k}\}_{k=0,...,M/2-1}\right) \nonumber \\
        &=&\sum_{a=1}^4 {\rm Var}\left(\{\tilde{\rho}^{(a)}_{m,j,k}\}_{k=0,...,M/2-1}\right) + \sum_{\substack{a,b=1,...,4 \\ a\ne b}} {\rm Cov}\left(\{\tilde{\rho}^{(a)}_{m,j,k}\}_{k=0,...,M/2-1},\{\tilde{\rho}^{(b)}_{m,j,k}\}_{k=0,...,M/2-1}\right) \nonumber \\
        &\le&\sum_{a=1}^4 {\rm Var}\left(\{\tilde{\rho}^{(a)}_{m,j,k}\}_{k=0,...,M/2-1}\right) + \sum_{\substack{a,b=1,...,4 \\ a\ne b}} \sqrt{{\rm Var}\left(\{\tilde{\rho}^{(a)}_{m,j,k}\}_{k=0,...,M/2-1}\right){\rm Var}\left(\{\tilde{\rho}^{(b)}_{m,j,k}\}_{k=0,...,M/2-1}\right)} \,,
    \end{eqnarray}
    we obtain Eq.~(\ref{eq:varrho}).
\end{proof}

\subsubsection{Main result}

Then, the following is our main result, a new quantum algorithm for GW matched filtering and a theorem on its query complexity and the number of qubits used.

\begin{theorem}
    Under Assumptions \ref{ass:oracle} and \ref{ass:hVar}, consider Problem \ref{prob:GWMF}.
	Let $\rho_{\rm soft}$ and $\rho_{\rm hard}$ be real numbers such that $0<\frac{\rho_{\rm hard}-\rho_{\rm soft}}{8\sqrt{5M}\gamma}<4$, and $\delta$ be a real number in $(0,1)$.
	Then, there is a quantum algorithm that uses
	\begin{equation}
	    O\left(\left(\log M+\log\left(\frac{\sqrt{M}\gamma}{\rho_{\rm hard}-\rho_{\rm soft}}\right)\right)\log\left(\frac{\sqrt{M}\gamma}{\rho_{\rm hard}-\rho_{\rm soft}}\right)\log\log\left(\frac{\sqrt{M}\gamma}{\rho_{\rm hard}-\rho_{\rm soft}}\right)\log\left(\frac{N_{\rm temp}M}{\delta}\right)\right) \label{eq:qnum}
	\end{equation}
	qubits and behaves as follows:
	\begin{itemize}
	    \item The algorithm outputs either of
	    \begin{enumerate}
	    \renewcommand{\labelenumi}{(\Alph{enumi})}
	        \item a message ``there is a signal" and an integer pair $(m,j)\in[N_{\rm temp}]_0\times[M]_0$ such that $\rho_{m,j}\ge\rho_{\rm soft}$,
	        \item a message ``there is no signal".
	    \end{enumerate}

	    \item If $r(\rho_{\rm hard})>0$, the algorithm outputs (A) with probability at least $1-\delta$.
	    In the algorithm, the number of queries to $O_{\rm Re}$ and $O_{\rm Im}$ is of order
	    \begin{equation}
	    O\left(\frac{\gamma\sqrt{M}}{(\rho_{\rm hard}-\rho_{\rm soft})\sqrt{\tilde{r}(\rho_{\rm hard})}}\log^{3/2}\left(\frac{\gamma\sqrt{M}}{\rho_{\rm hard}-\rho_{\rm soft}}\right)\log\log\left(\frac{\gamma\sqrt{M}}{\rho_{\rm hard}-\rho_{\rm soft}}\right)\log\left(\frac{N_{\rm temp}M}{\delta}\right)\log \delta^{-1}\right), \label{eq:compSigCase1}
	    \end{equation}
	    and thus
	    \begin{equation}
		O\left(\frac{\gamma M}{(\rho_{\rm hard}-\rho_{\rm soft})\sqrt{r(\rho_{\rm hard})}}\log^{3/2}\left(\frac{\gamma \sqrt{M}}{\rho_{\rm hard}-\rho_{\rm soft}}\right)\log\log\left(\frac{\gamma\sqrt{M}}{\rho_{\rm hard}-\rho_{\rm soft}}\right)\log\left(\frac{N_{\rm temp}M}{\delta}\right)\log \delta^{-1}\right). \label{eq:compSigCase2}
	    \end{equation}
	    
    \item If $r(\rho_{\rm soft})=0$, the algorithm outputs (B) with certainty.
    In the algorithm, the number of queries to $O_{\rm Re}$ and $O_{\rm Im}$ is of order
    \begin{equation}
	O\left(\frac{\gamma M\sqrt{N_{\rm temp}}}{\rho_{\rm hard}-\rho_{\rm soft}}\log^{3/2}\left(\frac{\gamma \sqrt{M}}{\rho_{\rm hard}-\rho_{\rm soft}}\right)\log\log\left(\frac{\gamma\sqrt{M}}{\rho_{\rm hard}-\rho_{\rm soft}}\right)\log\left(\frac{N_{\rm temp}M}{\delta}\right)\log \delta^{-1}\right). \label{eq:compNoSigCase}
    \end{equation}
	    
    \item If $r(\rho_{\rm hard})=0$ and $r(\rho_{\rm soft})>0$, the algorithm outputs either (A) or (B).
    In the algorithm, the number of queries to $O_{\rm Re}$ and $O_{\rm Im}$ is of order as in Eq.~(\ref{eq:compNoSigCase}).
	\end{itemize}
	\label{th:main}
\end{theorem}

\begin{proof}
    We first present the algorithm, and then prove the statements on the query complexity and the qubit number.\\
    
    \noindent \textbf{Algorithm}
    
    The algorithm is shown in Algorithm \ref{alg:main}.
    Note that, because of Lemma \ref{lem:varrho}, we can set $\sigma$, the upper bound on the variance of $\{\tilde{\rho}_{m,j,k}\}_{k=0,...,M/2-1}$, as in line 1.
    
    \begin{algorithm}[htp]
	\caption{Proposed algorithm for GW matched filtering}
	\label{alg:main}
	\begin{algorithmic}[1]

		\STATE Set $\epsilon=\frac{\rho_{\rm hard}-\rho_{\rm soft}}{2}$, $\sigma=4\sqrt{5M}\gamma$ and $\delta^\prime=\frac{\delta}{4N_{\rm temp}M}$.
		
		\STATE Combining $O_{\rm Re}$, $O_{\rm Im}$ and some arithmetic oracles, construct an oracles $O_{\rho}$ such that, for every $(m,j,k)\in[N_{\rm temp}]_0\times\left[M\right]_0\times\left[\frac{M}{2}\right]_0$,
		\begin{equation}
		O_{\rho}\ket{m}\ket{j}\ket{k}\ket{0}=\ket{m}\ket{j}\ket{k}\ket{\tilde{\rho}_{m,j,k}}.
		\end{equation}
		
		\STATE On the basis of Theorem \ref{th:QMCI}, using $O_{\rho}$, construct an oracle $O_{\rho,\epsilon,\delta^\prime,\sigma}^{\rm mean}$ such that, for every $(m,j)\in[N_{\rm temp}]_0\times\left[M\right]_0$,
		\begin{equation}
		    O_{\rho,\epsilon,\delta^\prime,\sigma}^{\rm mean}\ket{m}\ket{j}\ket{0}=\ket{m}\ket{j}\sum_{y\in \mathcal{Y}_{m,j}} \alpha_{\rho,y}\ket{y},
		\end{equation}
        where $\mathcal{Y}_{m,j}$ is a finite set of real numbers that includes a subset $\tilde{\mathcal{Y}}_{m,j}$ consisting of $\epsilon$-approximations of $\rho_{m,j}$ and $\{\alpha_y\}_{y\in\mathcal{Y}_{m,j}}$ are complex numbers satisfying $\sum_{\tilde{y}\in\tilde{\mathcal{Y}}_{m,j}}|\alpha_{\rho,\tilde{y}}|^2\ge 1-\delta$.
        
        \STATE Construct an oracle $O_{\rm AE}$ that performs the following operation
        \begin{eqnarray}
            \ket{0}\ket{0}\ket{0}\ket{0}
            &\rightarrow& \sqrt{\frac{1}{N_{\rm temp}M}} \sum_{m=0}^{N_{\rm temp}-1}\sum_{j=0}^{M-1} \ket{m}\ket{j}\ket{0}\ket{0} \nonumber \\
            &\rightarrow&
            \sqrt{\frac{1}{N_{\rm temp}M}} \sum_{m=0}^{N_{\rm temp}-1}\sum_{j=0}^{M-1}\ket{m}\ket{j}\sum_{y\in \mathcal{Y}_{m,j}}\alpha_{\rho,y}\ket{y}\ket{0}\nonumber \\
            &\rightarrow&\sqrt{\frac{1}{N_{\rm temp}M}} \sum_{m=0}^{N_{\rm temp}-1}\sum_{j=0}^{M-1}\ket{m}\ket{j}\sum_{y\in \mathcal{Y}_{m,j}}\alpha_{\rho,y}\ket{y}\left(\mathbbm{1}_{y\ge\rho_{\rm mid}}\ket{1}+\mathbbm{1}_{y<\rho_{\rm mid}}\ket{0}\right), \label{eq:OAE}
        \end{eqnarray}
        where $\rho_{\rm mid}:=\frac{\rho_{\rm hard}+\rho_{\rm soft}}{2}$.
        In Eq.~(\ref{eq:OAE}), we use $O^{\rm EqPr}_{N_{\rm temp}}$ and $O^{\rm EqPr}_{M}$ at the first arrow, $O_{\rho,\epsilon,\delta^\prime,\sigma}^{\rm mean}$ at the second arrow, and a comparer with $\ket{\rho_{\rm mid}}$ on an undisplayed ancillary register at the last arrow.
        
        \STATE Run $\proc{QAA}\left(O_{\rm AE},\frac{1}{2N_{\rm temp}M},\frac{\delta}{2}\right)$.
        
        \IF{we get the message ``failure"}
            \STATE Output the message "there is no signal".
        \ELSE 
            \STATE Measure the first two registers in the quantum state output by QAA and let the outcome be $(m,j)$.
            \STATE Calculate $\rho_{m,j}$ classically.
            \IF{$\rho_{m,j}\ge\rho_{\rm soft}$}
                \STATE Output the message "there is a signal" and $(m,j)$.
            \ELSE
                \STATE Output the message "there is no signal".
            \ENDIF
        \ENDIF

	\end{algorithmic}
    \end{algorithm}

    \ \\
    
    \noindent \textbf{Query complexity and qubit number}
    
    We describe this part in \ref{sec:remMainProof} since it is rather technical.
    
\end{proof}

\subsection{Remarks on settings and assumptions}

Now, we discuss the validity on the settings and assumptions in the algorithm proposed above.

\subsubsection{Setting of two thresholds \label{sec:twoThres}}

In Algorithm \ref{alg:main}, we set the two SNR thresholds, $\rho_{\rm hard}$ and $\rho_{\rm soft}$, whose meanings are explained in Section \ref{sec:idea}.
We consider the following is a plausible way to set them.
First, let us denote a common value of the SNR threshold as $\rho_{\rm com}$, e.g. $\rho_{\rm com}=8$
\footnote{In the absence of the astrophysical signal, SNR follows the Rayleigh distribution. The false alarm probability of each template is given by $p_\mathrm{fa,temp} = \exp[-\rho_\mathrm{com}^2/2]$.
The resolution of the coalescence time $\Delta t_\mathrm{start}$ is determined by the mismatch between two waveforms having slightly different coalescence times. If we set the mismatch is 5\% and the waveforms is monochromatic with the frequency of 100Hz, we get $\Delta t_\mathrm{start} \sim O(10^{-3})$ sec.
Then, for the observation period of $T_\mathrm{obs}$, the expected number of false alarm events is estimated by $N_\mathrm{fa} \sim p_\mathrm{fa,temp} N_\mathrm{temp} T_\mathrm{obs} / \Delta t_\mathrm{start}$. Assuming $T_\mathrm{obs}=10^7 \mathrm{sec}$, $\Delta t_\mathrm{start} = 10^{-3} \mathrm{sec}$, and $N_\mathrm{temp} = 10^6$, we should set the SNR threshold to $\rho_\mathrm{comp} \sim \sqrt{2 \ln (10^{10} \times N_\mathrm{temp})} \sim 8.6$ if we suppress $N_\mathrm{fa}$ to $O(1)$. See Ref.~\cite{Cutler:1992tc} and Chap.7 of Ref.~\cite{Maggiore:2007ulw}}
\cite{abbott2016}.
Note that the calculated SNR, in general, has fluctuations due to the random detector noise, and its variance is $1$ under the current normalization of the templates. Thus, even without the QMCI error, events with $\rho\in[\rho_{\rm com}, \rho_{\rm com} + 1)$ could have a true SNR value smaller than the threshold, while events with $\rho > \rho_{\rm com} + 1$ are very likely to exceed the threshold.
In light of this, it is reasonable to set $\rho_{\rm soft}=\rho_{\rm com}$ and $\rho_{\rm hard}=\rho_{\rm com}+1$, that is, $\rho_{\rm soft}=8$ and $\rho_{\rm hard}=9$ for $\rho_{\rm com}=8$.
In this setting, Algorithm \ref{alg:main} detects events with $\rho\ge9$ with high probability 
and never falsely alarms us for events with $\rho<8$, and events with $\rho\in[8,9)$ are detected or missed depending on fluctuations by the detector noise and the QMCI error.

\subsubsection{Implementation of $O_{\rm Re}$ and $O_{\rm Im}$ \label{sec:oracle}}

Here, we discuss the validity of Assumption \ref{ass:oracle}, that is, implementability of $O_{\rm Re}$ and $O_{\rm Im}$.
If we have accesses to the following oracles $O_{\rm hSRe}$, $O_{\rm hSIm}$, $O_{\rm QRe}$ and $O_{\rm QIm}$ such that, for every $(m,k)\in[N_{\rm temp}]_0\times\left[\frac{M}{2}-1\right]$,
\begin{equation}
    O_{\rm hSRe}\ket{k}\ket{0}=\ket{k}\Ket{\Re\left(\frac{\tilde{h}(f_k)}{S_{\rm n}(f_k)}\right)},O_{\rm hSIm}\ket{k}\ket{0}=\ket{k}\Ket{\Im\left(\frac{\tilde{h}(f_k)}{S_{\rm n}(f_k)}\right)} 
\end{equation}
and
\begin{equation}
    O_{\rm QRe}\ket{m}\ket{k}\ket{0}=\ket{m}\ket{k}\Ket{\Re\tilde{Q}_m(f_k)},O_{\rm QIm}\ket{m}\ket{k}\ket{0}=\ket{m}\ket{k}\Ket{\Im\tilde{Q}_m(f_k)},
\end{equation}
we can combine these along with arithmetic oracles to construct $O_{\rm Re}$ and $O_{\rm Im}$ (note that the remaining factor $2/\Delta t$ is just a known real number independent of $m$ and $k$).

$O_{\rm QRe}$ and $O_{\rm QIm}$ are in fact implementable.
To see this, note that $\tilde{Q}_m(f_k)$ is given as an explicit function of intrinsic parameters and $f_k=k/T$ by theories of GW sources such as compact binary
coalescences~\cite{balasubramanian1996,owen1996,owen1999,allen2012}.
Therefore, if we can relate the index $m$ to intrinsic parameter values by some elementary function, which is in fact possible under simple lattice-like template spacing such as \cite{owen1996}, we can write $\tilde{Q}_m$ as an explicit function of $m$ and $k$ and thus construct $O_{\rm QRe}$ and $O_{\rm QRe}$ using arithmetic oracles.

On the other hand, $\tilde{h}(f_k)/S_{\rm n}(f_k)$ is a factor determined by the experimental data and not represented by an explicit function.
We therefore resort to QRAM \cite{giovannetti2008}.
This enables us to access $N$ recorded data $x_k$ labeled by $k\in[N]_0$ and load a specified entry onto a register as
\begin{equation}
    \ket{k}\ket{0}\rightarrow\ket{k}\ket{x_k}
\end{equation}
in superposition in $O(\log N)$ time.
It takes $O(M)$ time to register $\left\{\Re\left(\frac{\tilde{h}(f_k)}{S_{\rm n}(f_k)}\right)\right\}_{k=1,...,\frac{M}{2}-1}$ and $\left\{\Im\left(\frac{\tilde{h}(f_k)}{S_{\rm n}(f_k)}\right)\right\}_{k=1,...,\frac{M}{2}-1}$ into a QRAM in advance of running Algorithm \ref{alg:main}, but this is expected to be less time-consuming than Algorithm \ref{alg:main} itself, which has $\widetilde{O}(\sqrt{N_{\rm temp}}M)$ query complexity.

\subsubsection{The mean and variance of detector outputs \label{sec:hvar}}

Here, we see the validity of Assumption \ref{ass:hVar} as follows.
First, from Eq.~(\ref{eq:noisePSdisc}), for $x_k:=2\Re \tilde{n}(f_k)/\sqrt{S_{\rm n}(f_k)T}$ and $y_k:=2\Im \tilde{n}(f_k)/\sqrt{S_{\rm n}(f_k)T}$ with $k\in\left[\frac{M}{2}-1\right]$, we see that $x_1,...,x_{\frac{M}{2}-1},y_1,...,y_{\frac{M}{2}-1}$ are independent standard normal variables.
We assume that $M\gg 1$ and thus the sample means and the sample variances of $\{x_k\}_k$ and $\{y_k\}_k$ are equal to the population means and the population variances, that is, 0 and 1, respectively.
We also assume the usual situation that the signal is much smaller than the noise.
More concretely, we assume that, for every $k\in\left[\frac{M}{2}-1\right]$, $|u_k|\ll 1$ and $|v_k|\ll 1$ hold, where $u_k:=2\Re \tilde{s}(f_k)/\sqrt{S_{\rm n}(f_k)T}$ and $v_k:=2\Im \tilde{s}(f_k)/\sqrt{S_{\rm n}(f_k)T}$, and thus that $|{\rm Mean}(\{u_k\}_{k=1,...,M/2-1})|\le 1$, $|{\rm Mean}(\{v_k\}_{k=1,...,M/2-1})|\le 1$, ${\rm Var}(\{u_k\}_{k=1,...,M/2-1})\le 1$, and ${\rm Var}(\{v_k\}_{k=1,...,M/2-1})\le 1$ hold. 
Under these assumptions, we can obtain
\begin{eqnarray}
    &&\left|{\rm Mean}\left(\left\{\frac{2\Re\tilde{h}(f_k)}{\sqrt{S_{\rm n}(f_k)T}}\right\}_{k=1,...,M/2-1}\right)\right|^2 \le \left(\left|{\rm Mean}(\{x_k\}_{k=1,...,M/2-1})\right|+\left|{\rm Mean}(\{u_k\}_{k=1,...,M/2-1})\right|\right)^2 
    \le 1 ,, \nonumber \\
    &&\left|{\rm Mean}\left(\left\{\frac{2\Im\tilde{h}(f_k)}{\sqrt{S_{\rm n}(f_k)T}}\right\}_{k=1,...,M/2-1}\right)\right|^2 \le \left(\left|{\rm Mean}(\{y_k\}_{k=1,...,M/2-1})\right|+\left|{\rm Mean}(\{v_k\}_{k=1,...,M/2-1})\right|\right)^2 
    \le 1 \,, \nonumber \\
    &&{\rm Var}\left(\left\{\frac{2\Re\tilde{h}(f_k)}{\sqrt{S_{\rm n}(f_k)T}}\right\}_{k=1,...,M/2-1}\right) \le \left(\sqrt{{\rm Var}(\{x_k\}_{k=1,...,M/2-1})}+\sqrt{{\rm Var}(\{u_k\}_{k=1,...,M/2-1})}\right)^2 
    \le 4 \,, \nonumber \\
    &&{\rm Var}\left(\left\{\frac{2\Im\tilde{h}(f_k)}{\sqrt{S_{\rm n}(f_k)T}}\right\}_{k=1,...,M/2-1}\right) \le \left(\sqrt{{\rm Var}(\{y_k\}_{k=1,...,M/2-1})}+\sqrt{{\rm Var}(\{v_k\}_{k=1,...,M/2-1})}\right)^2 
    \le 4 \,.
\end{eqnarray}

\subsubsection{Magnitude of $\gamma$ \label{sec:gamma}}

Here we will study the validity of the statement made in Problem~\ref{prob:GWMF}, that the value of $\gamma$, defined in Eq.~\eqref{eq:gamma_def}, is of order $O(1)$. This claim is supported by the fact that the template $Q$ is normalized as $(Q|Q)=1$, that is
\begin{equation}
	\frac{4}{M}\Re\left(\sum^{\frac{M}{2}-1}_{k=1} \frac{|\tilde{Q}_m(f_k)|^2}{S_{\rm n}(f_k) \Delta t}\right) = 1 \,.
	\label{eq:Q_normed_explicit}
\end{equation}

Since $\gamma$ is defined as the square root of the maximum summand in the sum of the left hand side of Eq.~\eqref{eq:Q_normed_explicit}, $\gamma$ can take values in the following range:

\begin{equation}
    \frac{1}{\sqrt{2-\frac{4}{M}}} \leq \gamma \leq \frac{\sqrt{M}}{2} \, ,
\label{eq:gamma_bounds}
\end{equation}

\noindent where the lower bound corresponds to the case in which all summands have the same value and the upper bound corresponds to the case where only one summand contributes. Because GW interferometers and signals will usually be broadband, we will be closer to the limit in which a significant fraction of the summands have similar values and so $\gamma$ will be of order $O(1)$. 

In Fig.~\ref{fig:gamma}, we show an explicit example of this, where we compute the value of $\gamma$ according to Eq.~\eqref{eq:gamma_def} for the compact binary coalescence (CBC) case, modeled using templates $Q$ computed with the $\texttt{IMRPhenomPv2}$ waveform \cite{Khan_2019}. The value of $\gamma$ only depends in the amplitude evolution of the waveform which is mostly depend on the component masses, parametrized via the total mass $M=m_1+m_2$ and the mass ratio $q=m_2/m_1$. We study total masses between $1 M_\odot$ and $300 M_\odot$, mass ratios between 0.2 and 1, and for this example, we set the spins to 0. For the noise PSD, $S_{\rm n}(f)$, we 
use the Advanced LIGO design sensitivity \cite{AdvLIGO_design}. We assume a sampling rate of $2048$Hz, and low and high frequency cutoffs of $20$Hz and $1024$Hz, respectively.

In Fig.~\ref{fig:gamma}, we can observe that $\gamma$ is of order $O(1)$ in all the parameter space studied. The minimum value of $\gamma$ is $2.50$ at $M = 126 M_\odot$, $q = 1.00$, which is the point where the SNR is most homogeneusly spread out across frequencies due to the location of the merger (where $|\tilde{Q}(f)|^2\propto f^{-4/3}$ instead of $|\tilde{Q}(f)|^2\propto f^{-7/3}$ as in the inspiral \cite{IMR_approx}) just before the most sensitive frequency range of the interferometer. On the other hand, the maximum value of $\gamma$ is 
4.66
at $M = 300 M_\odot$, $q = 0.20$. For the very largest masses, we observe that $\gamma$ tends to increase due to the fact that the higher the mass, the more the template is shifted towards smaller and smaller frequencies, until only the frequencies close to the low frequency cutoff contribute. Nonetheless, $\gamma$ still takes $O(1)$ values for all masses that can be expected to be seen by the ground-based detectors.

\begin{figure}[t!]
\begin{center}
\includegraphics[width=0.7\textwidth]{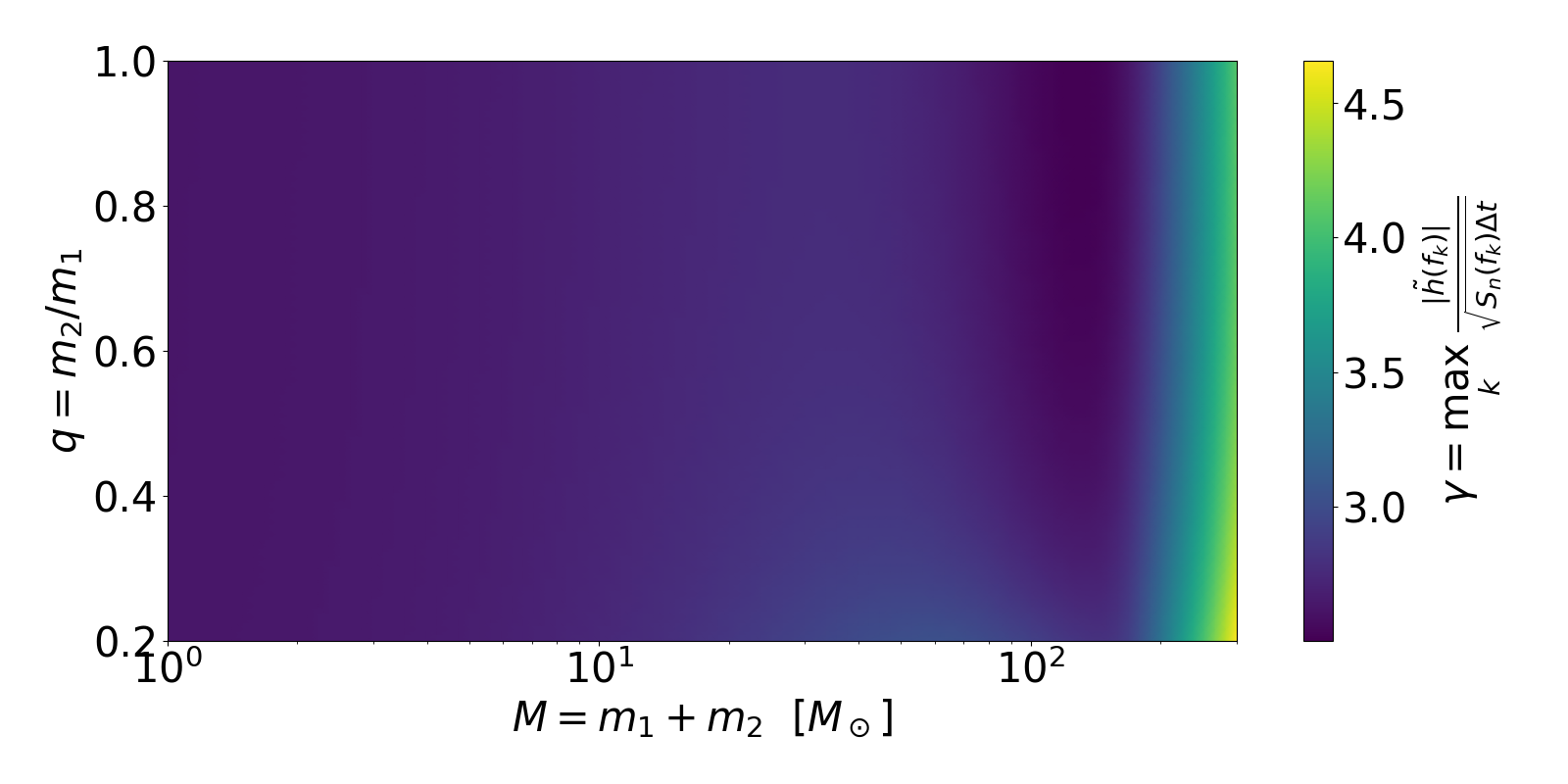}
\end{center} 
\caption{$\gamma$ computed with Eq.~\eqref{eq:gamma_def} using the $\texttt{IMRPhenomPv2}$ waveform \cite{Khan_2019} with spins set to 0 for the template $Q$ and the Advanced LIGO design sensitivity PSD for $S_{\rm n}(f)$ \cite{AdvLIGO_design}. We assume a sampling rate of $2048$Hz, and low and high frequency cutoffs of $20$Hz and $1024$Hz, respectively.}
\label{fig:gamma}
\end{figure}

\section{Summary \label{sec:summary}}

In this paper, we proposed a new quantum algorithm for GW matched filtering based on Ref.~\cite{gao2021} which has investigated the application of Grover’s search algorithm.
Our method, described in details in Sec.~\ref{sec:propAlg}, uses QMCI for the SNR calculation of Eq.~(\ref{eq:SNRdisc}) instead of FFT used in the classical method and in Ref.~\cite{gao2021}, and searches high SNR templates by QAA, running QMCI as a subroutine.
To deal with the erroneous nature of QMCI, we propose to set two thresholds $\rho_{\rm hard}$ and $\rho_{\rm soft}$ such that 
the proposed algorithm returns ``there is a signal" for events with SNR $\rho\ge\rho_{\rm hard}$ with high probability and ``there is no signal" for events with SNR $\rho<\rho_{\rm soft}$ with certainty. 

Our main results are summarized in Eqs.~(\ref{eq:compSigCase1}), (\ref{eq:compSigCase2}) and (\ref{eq:compNoSigCase}). By neglecting the  logarithmic dependencies, we can see that the proposed algorithm has $\widetilde{O}(M/\sqrt{r(\rho_{\rm th})})$ and $\widetilde{O}(M\sqrt{N_{\rm temp}})$ query complexity for $\tilde{r}(\rho_{\rm hard})>0$ and $\tilde{r}(\rho_{\rm hard})=0$, respectively. This still indicates a quadratic speedup with respect to the template number $N_{\rm temp}$ and the same order of complexity with respect to the number of time-series data points $M$ compared to the algorithm of Ref.~\cite{gao2021}, which is summarized in Sec.~\ref{sec:previous}. We note that the choice of $\rho_{\rm hard}$ and $\rho_{\rm soft}$, the accuracy of QMCI in other words, mildly affects the complexity, but according to the discussion in Sec.~\ref{sec:twoThres}, we take $\rho_{\rm hard}-\rho_{\rm soft}=1$, and thus not changing the factor.

The advantage of this algorithm is that it requires only qubit number that logarithmically scales on $M$ as described in Eq.~(\ref{eq:qnum}), contrary to the FFT-based method that requires $O(M)$ qubits.
Therefore, this algorithm is expected to be beneficial in the situation that quantum computers have a limitation on qubit number, which is likely to occur due to the large overhead for quantum error correction.
A possible drawback of the proposed algorithm compared to the algorithm in \cite{gao2021} is that, to load the detector output data onto a register in superposition in QMCI, it uses QRAM, whose experimental realization is challenging \cite{arunachalam2015}. 

In any case, we believe that proposing multiple ways of applying quantum methods that have different pros and cons is highly meaningful in taking advantage of quantum computing in future GW experiments, given today's uncertainty on what the future quantum computers will be.
Discussions for applying quantum computing in experimental physics have just started. We anticipate that more proposals will follow not only for GW data analysis but also for other heavy data analyses in various physical and astronomical experiments.

\section*{Acknowledgements}
G.M., S.K., and S.N. acknowledge support from the research project  PGC2018-094773-B-C32, and the Spanish Research Agency (Agencia Estatal de Investigaci\'on) through the Grant IFT Centro de Excelencia Severo Ochoa No CEX2020-001007-S, funded by MCIN/AEI/10.13039/501100011033.
S.K. is supported by the Spanish Atracci\'on de Talento contract no. 2019-T1/TIC-13177 granted by Comunidad de Madrid, the I+D grant PID2020-118159GA-C42 of the Spanish Ministry of Science and Innovation and the i-LINK 2021 grant LINKA20416 of CSIC.
T.Y. and S.K. are supported by Japan Society for the Promotion of Science (JSPS) KAKENHI Grant no. JP20H01899 and JP20H05853.
K.M. is supported by MEXT Quantum Leap Flagship Program (MEXT Q-LEAP) Grant no. JPMXS0120319794 and JSPS KAKENHI Grant no. JP22K11924.

\appendix
\renewcommand{\thesection}{Appendix \Alph{section}}

\section{Proof of Theorem \ref{th:QMCI} \label{sec:proofQMCI}}

The quantum algorithm for estimating the mean of variables with bounded variance is given as Algorithm 3 in \cite{montanaro2015}.
In order to use it as a subroutine in QAA, we now want to present it as a unitary transformation.
That is, we aim to remove measurements in the original algorithm in \cite{montanaro2015}.
Then, the modified algorithm outputs a quantum state in which the computational basis states corresponding to approximations of the mean have squared amplitudes summing up to almost 1.

We first present the following theorem on the mean estimation method for a bounded variable.
The method and the theorem are almost the same as Algorithm 1 and Theorem 2.3 in \cite{montanaro2015} in a specific case.
But, to be self-contained, we now present them with proof.
Note that, although Algorithm 1 in \cite{montanaro2015} is a procedure containing measurements, the following is a unitary transformation with no measurement to generate some quantum state.

\begin{theorem}
    Let $N\in\mathbb{N}$ and $\mathcal{X}$ be a set of $N$ real numbers $X_0,...,X_{N-1}\in[0,1]$.
    Suppose that we are given an oracle $O_\mathcal{X}$ that acts as Eq.~(\ref{eq:OX}). 
    Then, for any integer $t$ larger than 2, there is an oracle $O_{\mathcal{X},t}^{\rm mean}$ that acts as Eq.~(\ref{eq:ObarX}), where some ancillary qubits are undisplayed.
    Here, $\mathcal{Y}$ is a finite set of real numbers that includes a subset $\tilde{\mathcal{Y}}$ consisting of elements $\tilde{\mu}$ satisfying
    \begin{equation}
        |\tilde{\mu}-\mu|\le C\left(\frac{\sqrt{\mu}}{t}+\frac{1}{t^2}\right),
    \end{equation}
    with $\mu=\frac{1}{N}\sum_{i
    =0}^{N-1}X_i$ and a universal real constant $C$, and $\{\alpha_y\}_{y\in\mathcal{Y}}$ are complex numbers satisfying $\sum_{\tilde{y}\in\tilde{\mathcal{Y}}}|\alpha_{\tilde{y}}|^2\ge 8/\pi^2$.
    In $O_{\mathcal{X}}$, $O\left(t\right)$ queries to $O_X$ are made and $O\left(\log N + \log t\right)$ ancillary qubits are used.
\label{th:meanBd}
\end{theorem}

\begin{proof}
    
    Combining $O^{\rm EqPr}_N$ and $O_\mathcal{X}$, we can construct an oracle $O^\prime_\mathcal{X}$ on a system with $O(\log N)$ qubits in total such that
    \begin{equation}
        O^\prime_\mathcal{X} \ket{0}\ket{0}=\frac{1}{\sqrt{N}}\sum_{i=0}^{N-1} \ket{i}\ket{X_i}, \label{eq:OprimeX}
    \end{equation}
    and, combining this with some arithmetic oracles, we can construct an oracle $O^{\prime\prime}_\mathcal{X}$ that acts as
    \begin{equation}
        O^{\prime\prime}_\mathcal{X} \ket{0}\ket{0}\ket{0}=\frac{1}{\sqrt{N}}\sum_{i=0}^{N-1} \ket{i}\ket{X_i}\left(\sqrt{X_i}\ket{1}+\sqrt{1-X_i}\ket{0}\right),
    \end{equation}
    where the last ket corresponds to an ancillary qubit.
    Note that we can write the RHS as
    \begin{equation}
        \sqrt{\mu}\ket{\Phi_1}\ket{1}+\sqrt{1-\mu}\ket{\Phi_0}\ket{0},
    \end{equation}
    where $\ket{\Phi_1}$ and $\ket{\Phi_0}$ are some quantum states on the first two registers.
    Then, as stated in Theorem \ref{th:QAE}, using $O^{\prime\prime}_\mathcal{X}$ $O(t)$ times, we can construct $O_{\mathcal{X},t}^{\rm mean}$ that acts as
    \begin{equation}
    O_{\mathcal{X},t}^{\rm mean}\ket{0}=\sum_{y\in \mathcal{Y}} \alpha_y\ket{y}, \label{eq:OXt}
    \end{equation}
    where a real number set $\mathcal{Y}$ includes a subset $\tilde{\mathcal{Y}}$ such that, for every $\tilde{\mu}\in\tilde{\mathcal{Y}}$,
    \begin{equation}
    |\tilde{\mu}-\mu|\le\frac{2\pi\sqrt{\mu(1-\mu)}}{t}+\frac{\pi^2}{t^2}\le \pi^2\left(\frac{\sqrt{\mu}}{t}+\frac{1}{t^2}\right)
    \end{equation}
    holds and that $\sum_{\tilde{\mu}\in\tilde{\mathcal{Y}}}|\alpha_{\tilde{\mu}}|^2\ge8/\pi^2$.
    Since $O^{\prime\prime}_\mathcal{X}$ contains one query to $O_\mathcal{X}$ and uses $O(\log N)$ qubits, the statements on the query complexity and the qubit number immediately follows from Theorem \ref{th:QAE}.
\end{proof}

We then prove Theorem \ref{th:QMCI}.

\begin{proof}[Proof of Theorem \ref{th:QMCI}]
    Our Theorem \ref{th:QMCI} is almost the same as Theorem 2.5 in \cite{montanaro2015} in a specific case and we now just aim to modify Algorithm 3 in \cite{montanaro2015} so that all intermediate measurements are removed and that the output is a quantum state, see Eq.~(\ref{eq:ObarX}).
    Let us start by presenting the method with measurements as Algorithm \ref{alg:QMCIOrg}, which outputs an $\epsilon$-approximation of $\mu$, according to Theorem 2.5 in \cite{montanaro2015}.
    Note that, although the procedure looks different from that in Algorithm 3 in \cite{montanaro2015}, it is actually almost the same and the difference arises just because we explicitly write the steps that are originally shown separately as Algorithm 2 in \cite{montanaro2015}, in lines 7 and 8.
    There are only two differences, which enhance the lower bound of the success probability from $\frac{2}{3}$, the original value in \cite{montanaro2015}, to $1-\delta$.
    First, $K$, the number of repeated state generations and measurements in line 8, is different from that in Algorithm 2 in \cite{montanaro2015}.
    The current setting of $K$ makes the probability that, for each $(j,l)\in[J]\times[L+1]_0$, $\tilde{\tilde{\mu}}_{j,l}^+$ (resp. $\tilde{\tilde{\mu}}_{j,l}^-$) becomes an estimate of $\mu_{l,j}^+$ (resp. $\mu_{l,j}^-$) with desired accuracy larger than $1-\frac{5}{64(L+1)}$ (see Theorem \ref{th:median}).
    Thus, the probability that $2(L+1)$ estimations of $\mu_{j,0}^+,...,\mu_{j,L}^+,\mu_{j,0}^-,...,\mu_{j,L}^-$ simultaneously succeed is larger than $\left(1-\frac{5}{64(L+1)}\right)^{2(L+1)}\ge \frac{27}{32}$, and therefore, the probability that $\tilde{\mu}_j$ is $\epsilon$-close to $\mu$ is larger than $\frac{8}{9}\times \frac{27}{32}=\frac{3}{4}$ ($\frac{8}{9}$ is a lower bound on the probability for $\left|\sigma\tilde{m}_j-\mu\right|\le3\sigma$; see \cite{montanaro2015}).
    Second, Algorithm \ref{alg:QMCIOrg} in this paper has the loop on $j=1,...,J$, which means that Algorithm 3 in \cite{montanaro2015} itself is repeated $J$ times and that the median of the outputs is taken.
    This makes the lower bound of the success probability of the $\mu$ estimation from $\frac{3}{4}$ to $1-\delta$.
    
    \begin{algorithm}[htp]
	\caption{QMCI algorithm with measurements}
	\label{alg:QMCIOrg}
	\begin{algorithmic}[1]
		\REQUIRE{oracle $O_\mathcal{X}$ in Eq.~(\ref{eq:OX}), an upper bound $\sigma$ of ${\rm Var}(\mathcal{X})$, accuracy $\epsilon\in(0,4\sigma)$, and $\delta\in(0,1)$
		}
	
	    \STATE Set  $L:=\left\lceil\log_2\left(\frac{32\sigma}{\epsilon}\right)\right\rceil$, $J:=12\left\lceil\log\delta^{-1}\right\rceil+1$,
	    $K:=12\left\lceil\log\left(\frac{64(L+1)}{5}\right)\right\rceil+1$, $t_0:=\left\lceil\frac{32\sigma D\sqrt{\log_2\left(\frac{32\sigma}{\epsilon}\right)}}{\epsilon}\right\rceil$, where $D$ is a universal constant given in \cite{montanaro2015}.
	    
	    \FOR {$j=1,...,J$}
    		\STATE Randomly choose an integer $i_j$ from $[N]_0$ and generate the state $O_\mathcal{X}\ket{i_j}\ket{0}=\ket{i_j}\ket{X_{i_j}}$. Measure the second register in the computational basis and let the outcome divided by $\sigma$ be $\tilde{m}_j$.
		
	        \FOR {$l=0,1,...,L$}
	            \STATE Let $\tilde{\mathcal{X}}^{+}_{j,l}:=\{\tilde{X}^{+}_{l,1}(\tilde{m}_j),...,\tilde{X}^{+}_{l,N}(\tilde{m}_j)\}$ and $\tilde{\mathcal{X}}^{-}_{j,l}:=\{\tilde{X}^{-}_{l,1}(\tilde{m}_j),...,\tilde{X}^{-}_{l,N}(\tilde{m}_j)\}$, where, for $i\in[N]_0$ and $m\in\mathbb{R}$, $\tilde{X}^{\pm}_{l,i}(m)$ is given as
	            \begin{eqnarray}
	            \tilde{X}^{+}_{0,i}(m)&:=&
                \begin{cases}
                \frac{1}{4}\left(\frac{X_i}{\sigma}-\tilde{m}\right) & ; \ {\rm if} \ 0\le\frac{1}{4}\left(\frac{X_i}{\sigma}-\tilde{m}\right)<1 \\
                0 & ; \ {\rm otherwise}
                \end{cases}, \nonumber \\
		        \tilde{X}^{-}_{0,i}(m)&:=&
                \begin{cases}
                -\frac{1}{4}\left(\frac{X_i}{\sigma}-\tilde{m}\right) & ; \ {\rm if} \ -1<\frac{1}{4}\left(\frac{X_i}{\sigma}-\tilde{m}\right)<0 \\
                0 & ; \ {\rm otherwise}
                \end{cases}
		        \end{eqnarray}
		        when $l=0$, and
		        \begin{eqnarray}
		        \tilde{X}^{+}_{l,i}(m)&:=&
                \begin{cases}
                \frac{1}{4\cdot2^l}\left(\frac{X_i}{\sigma}-\tilde{m}\right) & ; \ {\rm if} \ 2^{l-1}\le\frac{1}{4}\left(\frac{X_i}{\sigma}-\tilde{m}\right)<2^l \\
                0 & ; \ {\rm otherwise}
                \end{cases}, \nonumber \\
		        \tilde{X}^{-}_{j,l,i}(m)&:=&
                \begin{cases}
                -\frac{1}{4\cdot2^l}\left(\frac{X_i}{\sigma}-\tilde{m}\right) & ; \ {\rm if} \ -2^{l}<\frac{1}{4}\left(\frac{X_i}{\sigma}-\tilde{m}\right)\le-2^{l-1} \\
                0 & ; \ {\rm otherwise}
                \end{cases}
		        \end{eqnarray}
		        when $l\ge1$.
		       
	            \STATE Construct an oracle $O_{\tilde{\mathcal{X}}^{+}_{j,l}}$ (resp. $O_{\tilde{\mathcal{X}}^{-}_{j,l}}$) such that $O_{\tilde{\mathcal{X}}^{+}_{j,l}}\ket{0}\ket{0}=\frac{1}{\sqrt{N}}\sum_{i=0}^{N-1}\ket{i}\ket{\tilde{X}^{+}_{l,i}(m_j)}$ (resp. $O_{\tilde{\mathcal{X}}^{-}_{j,l}}\ket{0}\ket{0}=\frac{1}{\sqrt{N}}\sum_{i=0}^{N-1}\ket{i}\ket{\tilde{X}^{-}_{l,i}(m_j)}$) by combining $O^{\rm EqPr}_N$, $O_\mathcal{X}$ and some arithmetic oracles.
	       
	            \STATE Using $O_{\tilde{\mathcal{X}}^{+}_{j,l}}$, construct an oracle $O^{\rm mean}_{\tilde{\mathcal{X}}^{+}_{j,l},t_0}$ that acts like Eq.~(\ref{eq:ObarX}), that is, $O^{\rm mean}_{\tilde{\mathcal{X}}^{+}_{j,l},t_0}\ket{0}=\sum_{y\in\mathcal{Y}^{+}_{j,l}}\alpha_y\ket{y}$, where a real number set $\mathcal{Y}^{+}_{j,l}$ includes a subset $\tilde{\mathcal{Y}}^{+}_{j,l}$ such that, for every $\tilde{\mu}\in\tilde{\mathcal{Y}}^{+}_{j,l}$, $|\tilde{\mu}-\mu_{j,l}^+|\le C\left(\frac{\sqrt{\mu_{j,l}^+}}{t_0}+\frac{1}{t_0^2}\right)$ holds with $\mu_{j,l}^+:=\frac{1}{N}\sum_{i=0}^{N-1}\tilde{X}^{+}_{j,l,i}$ and that $\sum_{\tilde{\mu}\in\tilde{\mathcal{Y}}^{+}_{j,l}}|\alpha_{\tilde{\mu}}|^2\ge8/\pi^2$.
	            
	            Similarly, using $O_{\tilde{\mathcal{X}}^{-}_{j,l}}$, construct an oracle $O^{\rm mean}_{\tilde{\mathcal{X}}^{-}_{j,l},t_0}$ that acts as $O^{\rm mean}_{\tilde{\mathcal{X}}^{-}_{j,l},t_0}\ket{0}=\sum_{y\in\mathcal{Y}^{-}_{j,l}}\alpha_y\ket{y}$, where a real number set $\mathcal{Y}^{-}_{j,l}$ includes a subset $\tilde{\mathcal{Y}}^{-}_{j,l}$ such that, for every $\tilde{\mu}\in\tilde{\mathcal{Y}}^{-}_{j,l}$, $|\tilde{\mu}-\mu_{j,l}^-|\le C\left(\frac{\sqrt{\mu_{j,l}^-}}{t_0}+\frac{1}{t_0^2}\right)$ holds with $\mu_{j,l}^-:=\frac{1}{N}\sum_{i=0}^{N-1}\tilde{X}^{-}_{j,l,i}$ and that $\sum_{\tilde{\mu}\in\tilde{\mathcal{Y}}^{-}_{j,l}}|\alpha_{\tilde{\mu}}|^2\ge8/\pi^2$. 
	       
	            \STATE Generate $K$ copies of the quantum state $O^{\rm mean}_{\tilde{\mathcal{X}}^{+}_{j,l},t_0}\ket{0}$ (resp. $O^{\rm mean}_{\tilde{\mathcal{X}}^{-}_{j,l},t_0}\ket{0}$) and measure them in the computational basis.
	            Let the median of the measurement outcomes be $\tilde{\tilde{\mu}}^+_{j,l}$ (resp. $\tilde{\tilde{\mu}}^-_{j,l}$).

	        \ENDFOR
		    
	        \STATE Set $\tilde{\mu}_j:=\sigma\left(\tilde{m}_j+4\sum_{l=0}^L2^l(\tilde{\tilde{\mu}}^+_{j,l}-\tilde{\tilde{\mu}}^-_{j,l})\right)$.
		    
		\ENDFOR
		
		\STATE Output the median of $\tilde{\mu}_1,...,\tilde{\mu}_J$.
		
	\end{algorithmic}
\end{algorithm}

Now, let us present the implementation of this algorithm without intermediate measurements.
For a preparation, we randomly choose $J$ integers from $[N]_0$ and let them $i_1,...,i_J$.
Then, on a four-register system initialized as $\ket{i_j}\ket{0}\ket{0}\ket{0}$, we perform $O_{\mathcal{X}}$ and a division to yield $\ket{i_j}\ket{\tilde{m}_j}\ket{0}\ket{0}$, where $\tilde{m}_j:=X_{i_j}/\sigma$.
Furthermore, combining $O^{\rm EqPr}_N$, $O_{\mathcal{X}}$ and some arithmetic oracles, we construct an oracle $O_{\tilde{\mathcal{X}}^{+}_{j,l}}$ that acts on $\ket{i_j}\ket{\tilde{m}_j}\ket{0}\ket{0}$ as
\begin{equation}
O_{\tilde{\mathcal{X}}^{+}_{j,l}}\ket{i_j}\ket{\tilde{m}_j}\ket{0}\ket{0}=\ket{i_j}\ket{\tilde{m}_j}\left(\frac{1}{\sqrt{N}}\sum_{i=0}^{N-1}\ket{i}\ket{\tilde{X}^{+}_{l,i}(\tilde{m}_j)}\right).
\end{equation}
Similarly, we obtain an oracle $O_{\tilde{\mathcal{X}}^{-}_{j,l}}$ that acts as
\begin{equation}
    O_{\tilde{\mathcal{X}}^{-}_{j,l}}\ket{i_j}\ket{\tilde{m}_j}\ket{0}\ket{0}=\ket{i_j}\ket{\tilde{m}_j}\left(\frac{1}{\sqrt{N}}\sum_{i=0}^{N-1}\ket{i}\ket{\tilde{X}^{-}_{l,i}(\tilde{m}_j)}\right).
\end{equation}
Using $O_{\tilde{\mathcal{X}}^{+}_{j,l}}$ and $O_{\tilde{\mathcal{X}}^{-}_{j,l}}$ $O(t_0)$ times, we can construct oracles $\tilde{O}^{\rm mean}_{\tilde{\mathcal{X}}^{+}_{j,l},t_0}$ and $\tilde{O}^{\rm mean}_{\tilde{\mathcal{X}}^{-}_{j,l},t_0}$, which resemble $O^{\rm mean}_{\tilde{\mathcal{X}}^{+}_{j,l},t_0}$ and $O^{\rm mean}_{\tilde{\mathcal{X}}^{-}_{j,l},t_0}$ in Algorithm \ref{alg:QMCIOrg}, respectively, but act as
\begin{equation}
\tilde{O}^{\rm mean}_{\tilde{\mathcal{X}}^{+}_{j,l},t_0}\ket{i_j}\ket{\tilde{m}_j}\ket{0}=\ket{i_j}\ket{\tilde{m}_j}\left(\sum_{y\in\mathcal{Y}^+_{j,l}}\alpha_y\ket{y}\right), \tilde{O}^{\rm mean}_{\tilde{\mathcal{X}}^{-}_{j,l},t_0}\ket{i_j}\ket{\tilde{m}_j}\ket{0}=\ket{i_j}\ket{\tilde{m}_j}\left(\sum_{y\in\mathcal{Y}^-_{j,l}}\alpha_y\ket{y}\right),
\end{equation}
where $\mathcal{Y}^{\pm}_{j,l}$ and $\alpha_y$ are described in Algorithm \ref{alg:QMCIOrg} and some registers are not displayed.
Then, on an appropriate number of registers, some of which are initialized to $\ket{i_1},...,\ket{i_J}$, we use $O_{\mathcal{X}}$, $\tilde{O}^{\rm mean}_{\tilde{\mathcal{X}}^{+}_{j,l},t_0}$ and $\tilde{O}^{\rm mean}_{\tilde{\mathcal{X}}^{-}_{j,l},t_0}$ to generate the following quantum state:
\begin{eqnarray}
    &&\bigotimes_{j=1}^J\ket{i_j}\ket{\tilde{m}_j}\otimes\left(\bigotimes_{l=0}^L \left(\sum_{y^+_{j,l,1}\in\mathcal{Y}^+_{j,l}}\alpha_{y^+_{j,l,1}}\ket{y^+_{j,l,1}}\right)\otimes\cdots\otimes\left(\sum_{y^+_{j,l,K}\in\mathcal{Y}^+_{j,l}}\alpha_{y^+_{j,l,K}}\ket{y^+_{j,l,K}}\right)\right. \nonumber \\
    &&\qquad\qquad\qquad\qquad\qquad\left.\otimes\left(\sum_{y^-_{j,l.1}\in\mathcal{Y}^-_{j,l}}\alpha_{y^-_{j,l,1}}\ket{y^-_{j,l,1}}\right)\otimes\cdots\otimes\left(\sum_{y^-_{j,l,K}\in\mathcal{Y}^-_{j,l}}\alpha_{y^-_{j,l,K}}\ket{y^-_{j,l,K}}\right)\right) \nonumber\\
    &=&\bigotimes_{j=1}^J\ket{i_j}\ket{\tilde{m}_j}\otimes\left(\bigotimes_{l=0}^L \sum_{\substack{y^+_{j,l,1}\in\mathcal{Y}^+_{j,l},\cdots,y^+_{j,l,K}\in\mathcal{Y}^+_{j,l} \\ y^-_{j,l,1}\in\mathcal{Y}^-_{j,l},\cdots,y^-_{j,l,K}\in\mathcal{Y}^-_{j,l} }}\left(\prod_{k=1}^K\alpha_{y^+_{j,l,k}}\alpha_{y^-_{j,l,k}}\right)\ket{y^+_{j,l,1}}\cdots\ket{y^+_{j,l,K}}\ket{y^-_{j,l,1}}\cdots\ket{y^-_{j,l,K}}\right). \label{eq:QMCIStateInt}
\end{eqnarray}
Further, adding some registers and performing $O^{\rm med}_K$, we obtain
\begin{equation}
    \bigotimes_{j=1}^J\ket{i_j}\ket{\tilde{m}_j}\otimes\left(\bigotimes_{l=0}^L \sum_{\substack{y^+_{j,l,1}\in\mathcal{Y}^+_{j,l},\cdots,y^+_{j,l,K}\in\mathcal{Y}^+_{j,l} \\ y^-_{j,l,1}\in\mathcal{Y}^-_{j,l},\cdots,y^-_{j,l,K}\in\mathcal{Y}^-_{j,l} }}\left(\prod_{k=1}^K\alpha_{y^+_{j,l,k}}\alpha_{y^-_{j,l,k}}\right)\ket{y^+_{j,l,1}}\cdots\ket{y^+_{j,l,K}}\ket{y^-_{j,l,1}}\cdots\ket{y^-_{j,l,K}}\ket{\tilde{\tilde{\mu}}^+_{j,l}}\ket{\tilde{\tilde{\mu}}^-_{j,l}}\right),
\end{equation}
where $\tilde{\tilde{\mu}}^+_{j,l}={\rm med}(y^+_{j,l,1},...,y^+_{j,l,K})$ and $\tilde{\tilde{\mu}}^+_{j,l}={\rm med}(y^-_{j,l,1},...,y^-_{j,l,K})$.
Moreover, adding further registers and using arithmetic oracles, we obtain
\begin{equation}
    \bigotimes_{j=1}^J \ket{i_j}\ket{\tilde{m}_j}\otimes\left(\left(\bigotimes_{l=0}^L\sum_{\substack{y^+_{j,l,1}\in\mathcal{Y}^+_{j,l},\cdots,y^+_{j,l,K}\in\mathcal{Y}^+_{j,l} \\ y^-_{j,l,1}\in\mathcal{Y}^-_{j,l},\cdots,y^-_{j,l,K}\in\mathcal{Y}^-_{j,l} }}\left(\prod_{k=1}^K\alpha_{y^+_{j,l,k}}\alpha_{y^-_{j,l,k}}\right)\ket{y^+_{j,l,1}}\cdots\ket{y^+_{j,l,K}}\ket{y^-_{j,l,1}}\cdots\ket{y^-_{j,l,K}}\ket{\tilde{\tilde{\mu}}^+_{j,l}}\ket{\tilde{\tilde{\mu}}^-_{j,l}}\right)\otimes\ket{\tilde{\mu}_j}\right),
\end{equation}
where $\tilde{\mu}_j:=\sigma\left(\tilde{m}_j+4\sum_{l=0}^L2^l(\tilde{\tilde{\mu}}^+_{j,l}-\tilde{\tilde{\mu}}^-_{j,l})\right)$.
Finally, performing $O^{\rm med}_J$ yields the state
\begin{equation}
    \left[\bigotimes_{j=1}^J \ket{i_j}\ket{\tilde{m}_j}\otimes\left(\left(\bigotimes_{l=0}^L\sum_{\substack{y^+_{j,l,1}\in\mathcal{Y}^+_{j,l},\cdots,y^+_{j,l,K}\in\mathcal{Y}^+_{j,l} \\ y^-_{j,l,1}\in\mathcal{Y}^-_{j,l},\cdots,y^-_{j,l,K}\in\mathcal{Y}^-_{j,l} }}\left(\prod_{k=1}^K\alpha_{y^+_{j,l,k}}\alpha_{y^-_{j,l,k}}\right)\ket{y^+_{j,l,1}}\cdots\ket{y^+_{j,l,K}}\ket{y^-_{j,l,1}}\cdots\ket{y^-_{j,l,K}}\ket{\tilde{\tilde{\mu}}^+_{j,l}}\ket{\tilde{\tilde{\mu}}^-_{j,l}}\right)\otimes\ket{\tilde{\mu}_j}\right)\right]\otimes\ket{\tilde{\mu}}, \label{eq:QMCIFinState}
\end{equation}
where $\tilde{\mu}={\rm med}(\tilde{\mu}_1,...,\tilde{\mu}_J)$.
As Algorithm \ref{alg:QMCIOrg}, if we measure the last register in this final state, we obtain an $\epsilon$-approximation of $\mu$ with probability at least $1-\delta$.
This means that the final state, Eq.~(\ref{eq:QMCIFinState}), can be written as Eq.~(\ref{eq:ObarX}), with some registers undisplayed.
We therefore regard the above unitary transformation to generate Eq.~(\ref{eq:QMCIFinState}) as $O_{\mathcal{X},\epsilon,\delta,\sigma}^{\rm mean}$.

Lastly, let us consider the statements on the query complexity and the qubit number.
Note that we generate the state of Eq.~(\ref{eq:QMCIStateInt}) by $O(J)$ uses of $O_{\mathcal{X}}$ and
\begin{equation}
O(JKL)=O\left(\log\left(\frac{\sigma}{\epsilon}\right)\log\log\left(\frac{\sigma}{\epsilon}\right)\log\delta^{-1}\right) \label{eq:JKL}
\end{equation}
uses of $\tilde{O}^{\rm mean}_{\tilde{\mathcal{X}}^{+}_{j,l},t_0}$ and $\tilde{O}^{\rm mean}_{\tilde{\mathcal{X}}^{-}_{j,l},t_0}$, along with initialization of some registers to $\ket{i_1},...,\ket{i_J}$, and 
the transformation from Eq.~(\ref{eq:QMCIStateInt}) to Eq.~(\ref{eq:QMCIFinState}) is done by only arithmetic oracles.
Also note that, in $\tilde{O}^{\rm mean}_{\tilde{\mathcal{X}}^{+}_{j,l},t_0}$ and $\tilde{O}^{\rm mean}_{\tilde{\mathcal{X}}^{-}_{j,l},t_0}$, $O_{\tilde{\mathcal{X}}^{+}_{j,l}}$ and $O_{\tilde{\mathcal{X}}^{+}_{j,l}}$ are called $O(t_0)=O\left(\frac{\sigma}{\epsilon}\log^{1/2}\left(\frac{\sigma}{\epsilon}\right)\right)$ times.
The number of calls to $O_{\mathcal{X}}$ in $\tilde{O}^{\rm mean}_{\tilde{\mathcal{X}}^{+}_{j,l},t_0}$ and $\tilde{O}^{\rm mean}_{\tilde{\mathcal{X}}^{-}_{j,l},t_0}$ is also of the same order, since each of $O_{\tilde{\mathcal{X}}^{+}_{j,l}}$ and $O_{\tilde{\mathcal{X}}^{+}_{j,l}}$ contains one call to $O_{\mathcal{X}}$.
Combining these observations, we see that the number of uses of $O_{\mathcal{X}}$ in generating Eq.~(\ref{eq:QMCIFinState}) is given by  Eq.~(\ref{eq:compQMCI}).
When it comes to qubit number, we note that, as stated in Theorem \ref{th:meanBd}, $\tilde{O}^{\rm mean}_{\tilde{\mathcal{X}}^{+}_{j,l},t_0}$ and $\tilde{O}^{\rm mean}_{\tilde{\mathcal{X}}^{-}_{j,l},t_0}$ use $O(\log N+\log t_0)=O\left(\log N+\log \left(\frac{\sigma}{\epsilon}\right)\right)$ qubits, and thus we use $O\left(\left(\log N+\log \left(\frac{\sigma}{\epsilon}\right)\right)\times JKL\right)$ qubits, which is of order as in Eq.~(\ref{eq:qubitQMCI}), in preparing Eq.~(\ref{eq:QMCIStateInt}).
Also note that added registers in transformation from Eq.~(\ref{eq:QMCIStateInt}) to Eq.~(\ref{eq:QMCIFinState}) is $O(JL)$.
From these observations, the total number of qubits used in generating Eq.~(\ref{eq:QMCIFinState}) is given by Eq.~(\ref{eq:qubitQMCI}).

\end{proof}

\section{Remaining part of the proof of Theorem \ref{th:main} \label{sec:remMainProof}}

\begin{proof}[The remaining part of the proof]
 \ \\
  \ \\
    \noindent \textbf{Query complexity and qubit number}\\
    
    We consider the following cases separately.\\
    
    \noindent (i) $r(\rho_{\rm hard})>0$\\
    
    In this case, there exists $(m,j)\in[N_{\rm temp}]_0\times\left[M\right]_0$ such that $\rho_{m,j}\ge\rho_{\rm hard}$.
    For such $(m,j)$,
    \begin{equation}
        |y-\rho_{m,j}|\le\epsilon\Rightarrow y\ge\rho_{\rm mid}
    \end{equation}
    holds for any $y\in\mathbb{R}$ (recall that we are now setting $\epsilon=\frac{\rho_{\rm hard}-\rho_{\rm soft}}{2}$), and thus
    \begin{equation}
        \sum_{\substack{y\in \mathcal{Y}_{m,j} \\ y\ge\rho_{\rm mid}}}|\alpha_{\rho,y}|^2\ge \sum_{\substack{y\in \mathcal{Y}_{m,j} \\ |y-\rho_{m,j}|\le\epsilon}}|\alpha_{\rho,y}|^2\ge 1-\delta^\prime\ge \frac{1}{2}
    \end{equation}
    holds.
    Using this, $p_1$, the probability that we obtain 1 when we measure the last qubit in the final state in Eq.~(\ref{eq:OAE}), is evaluated as
    \begin{equation}
        p_1=\frac{1}{N_{\rm temp}M}\sum_{m=0}^{N_{\rm temp}-1}\sum_{j=0}^{M-1}\sum_{\substack{y\in \mathcal{Y}_{m,j} \\ y\ge\rho_{\rm mid}}}|\alpha_{\rho,y}|^2\ge\frac{1}{N_{\rm temp}M}\sum_{\substack{(m,j)\in[N_{\rm temp}]_0\times\left[M\right]_0\\ \rho_{m,j}\ge\rho_{\rm hard}}}\frac{1}{2}=\frac{\tilde{r}(\rho_{\rm hard})}{2}\ge\frac{1}{2N_{\rm temp}M}.
    \end{equation}
    Therefore, by $\proc{QAA}\left(O_{\rm AE},\frac{1}{2N_{\rm temp}M},\frac{\delta}{2}\right)$, we get the state
    \begin{equation}
        \ket{\psi}:=\frac{1}{\sqrt{p_1}}\sum_{m=0}^{N_{\rm temp}-1}\sum_{j=0}^{M-1}\ket{m}\ket{j}\sum_{\substack{y\in \mathcal{Y}_{m,j} \\ y\ge\rho_{\rm mid}}}\alpha_{\rho,y}\ket{y},
    \end{equation}
    with probability at least $1-\frac{\delta}{2}$.
    
    On the other hand, note that, for $(m,j)\in[N_{\rm temp}]_0\times\left[M\right]_0$ such that $\rho_{m,j}<\rho_{\rm soft}$, 
    \begin{equation}
        y\ge\rho_{\rm mid}\Rightarrow|y-\rho_{m,j}|>\epsilon\Rightarrow y\notin\tilde{\mathcal{Y}}_{m,j}
    \end{equation}
    holds for any $y\in\mathbb{R}$, and thus we have
    \begin{equation}
        \sum_{\substack{y\in \mathcal{Y}_{m,j} \\ y\ge\rho_{\rm mid}}}|\alpha_{\rho,y}|^2 = \sum_{\substack{y\in \mathcal{Y}_{m,j}\setminus\tilde{\mathcal{Y}}_{m,j} \\ y\ge\rho_{\rm mid}}}|\alpha_{\rho,y}|^2\le \sum_{y\in \mathcal{Y}_{m,j}\setminus\tilde{\mathcal{Y}}_{m,j}}|\alpha_{\rho,y}|^2 < \delta^\prime \,.
    \end{equation}
    This means that $p_{\ge \rho_{\rm soft}}$ the probability that we obtain $(m,j)$ such that $\rho_{m,j}\ge\rho_{\rm soft}$ when we measure the first two registers in $\ket{\psi}$ is evaluated as
    \begin{eqnarray}
        p_{\ge \rho_{\rm soft}} &=& \frac{1}{p_1}\sum_{\substack{(m,j)\in[N_{\rm temp}]_0\times\left[M\right]_0 \\ \rho_{m,j}\ge\rho_{\rm soft}}}\sum_{\substack{y\in \mathcal{Y}_{m,j} \\ y\ge\rho_{\rm mid}}}\frac{1}{N_{\rm temp}M}|\alpha_{\rho,y}|^2 \nonumber \\
        &=&\frac{\sum_{\substack{(m,j)\in[N_{\rm temp}]_0\times\left[M\right]_0 \\ \rho_{m,j}\ge\rho_{\rm soft}}}\sum_{\substack{y\in \mathcal{Y}_{m,j} \\ y\ge\rho_{\rm mid}}}|\alpha_{\rho,y}|^2}{\sum_{(m,j)\in[N_{\rm temp}]_0\times\left[M\right]_0}\sum_{\substack{y\in \mathcal{Y}_{m,j} \\ y\ge\rho_{\rm mid}}}|\alpha_{\rho,y}|^2} \nonumber \\
        &=& 1- \frac{\sum_{\substack{(m,j)\in[N_{\rm temp}]_0\times\left[M\right]_0 \\ \rho_{m,j}<\rho_{\rm soft}}}\sum_{\substack{y\in \mathcal{Y}_{m,j} \\ y\ge\rho_{\rm mid}}}|\alpha_{\rho,y}|^2}{\sum_{(m,j)\in[N_{\rm temp}]_0\times\left[M\right]_0}\sum_{\substack{y\in \mathcal{Y}_{m,j} \\ y\ge\rho_{\rm mid}}}|\alpha_{\rho,y}|^2} \nonumber \\
        &\ge& 1-2N_{\rm temp}M\delta^\prime \nonumber \\
        &=& 1-\frac{\delta}{2},
    \end{eqnarray}
    where, at the inequality, we use
    \begin{equation}
        \sum_{\substack{(m,j)\in[N_{\rm temp}]_0\times\left[M\right]_0 \\ \rho_{m,j}<\rho_{\rm soft}}}\sum_{\substack{y\in \mathcal{Y}_{m,j} \\ y\ge\rho_{\rm mid}}}|\alpha_{\rho,y}|^2\le \sum_{\substack{(m,j)\in[N_{\rm temp}]_0\times\left[M\right]_0 \\ \rho_{m,j}<\rho_{\rm soft}}}\delta^\prime\le\sum_{(m,j)\in[N_{\rm temp}]_0\times\left[M\right]_0}\delta^\prime=N_{\rm temp}M\delta^\prime
    \end{equation}
    and
    \begin{equation}
        \sum_{(m,j)\in[N_{\rm temp}]_0\times\left[M\right]_0}\sum_{\substack{y\in \mathcal{Y}_{m,j} \\ y\ge\rho_{\rm mid}}}|\alpha_{\rho,y}|^2 =N_{\rm temp}Mp_1\ge\frac{1}{2}.
    \end{equation}
    In summary, by the algorithm, we get $(m,j)$ such that $\rho_{m,j}\ge\rho_{\rm soft}$ with probability at least $\left(1-\frac{\delta}{2}\right)^2\ge1-\delta$.
    
    The query complexity is evaluated as follows.
    Until we get an output, $\proc{QAA}\left(O_{\rm AE},\frac{1}{2N_{\rm temp}M},\frac{\delta}{2}\right)$ calls $O_{\rm AE}$
    \begin{equation}
    O\left(\frac{\log \delta^{-1}}{\sqrt{p_1}}\right)=O\left(\frac{\log \delta^{-1}}{\sqrt{\tilde{r}(\rho_{\rm hard})}}\right)
    \end{equation}
    times.
    As stated in Theorem \ref{th:QMCI}, $O_{\rm AE}$ makes
    \begin{equation}
        O\left(\frac{\sigma}{\epsilon}\log^{3/2}\left(\frac{\sigma}{\epsilon}\right)\log\log\left(\frac{\sigma}{\epsilon}\right)\log\left(\frac{1}{\delta^\prime}\right)\right)=O\left(\frac{\sqrt{M}\gamma}{\rho_{\rm hard}-\rho_{\rm soft}}\log^{3/2}\left(\frac{\sqrt{M}\gamma}{\rho_{\rm hard}-\rho_{\rm soft}}\right)\log\log\left(\frac{\sqrt{M}\gamma}{\rho_{\rm hard}-\rho_{\rm soft}}\right)\log\left(\frac{N_{\rm temp}M}{\delta}\right)\right) \label{eq:numOrho}
    \end{equation}
    calls to $O_\rho$.
    The number of calls to $O_{\rm Re}$ and $O_{\rm Im}$ is of the same order, since $O_\rho$ contains $O(1)$ calls to them. 
    Combining these observations, we obtain the estimations given by Eqs.~(\ref{eq:compSigCase1}) and (\ref{eq:compSigCase2}) for the number of calls to $O_{\rm Re}$ and $O_{\rm Im}$ in the algorithm.
    
    The number of qubits used in this algorithm is dominated by that required to perform $O_{\rm AE}$ in Eq.~(\ref{eq:OAE}), since QAA does not require additional qubits.
    The number of qubits used for the operation in Eq.~(\ref{eq:OAE}) is estimated as follows.
    The first two registers have
    \begin{equation}
    O(\log N_{\rm temp}+\log M) \label{eq:qbitnum1st2nd}
    \end{equation}
    qubits in total.
    Besides, according to Theorem \ref{th:QMCI}, the third register and ancillary registers used for $O_{\rho,\epsilon,\delta^\prime,\sigma}^{\rm mean}$ at the second arrow have
    \begin{equation}
    O\left(\left(\log M+\log\left(\frac{\sigma}{\epsilon}\right)\right)\log\left(\frac{\sigma}{\epsilon}\right)\log\log\left(\frac{\sigma}{\epsilon}\right)\log\left(\frac{1}{\delta^\prime}\right)\right) \label{eq:qbitnum3rd}
    \end{equation}
    qubits in total.
    Eq.~(\ref{eq:qbitnum3rd}) becomes Eq.~(\ref{eq:qnum}) under the setting on $\epsilon$, $\sigma$ and $\delta^\prime$ in Algorithm \ref{alg:main}.
    Since Eq.~(\ref{eq:qbitnum1st2nd}) is subdominant to Eq.~(\ref{eq:qnum}), we have an upper bound on the total qubit number as Eq.~(\ref{eq:qnum}).\\
    
    \noindent (ii) $r(\rho_{\rm soft})=0$\\
    
    If $\proc{QAA}\left(O_{\rm AE},\frac{1}{2N_{\rm temp}M},\frac{\delta}{2}\right)$ outputs ``failure", the algorithm outputs ``there is no signal".
    In case QAA outputs some quantum state $\ket{\psi}$ by error, the algorithm goes to the second step where $\rho_{m,j}$ is classically calculated for $(m,j)\in[N_{\rm temp}]_0\times\left[M\right]_0$ given by the first two registers in $\ket{\psi}$.
    The output of the classical computation should be smaller than $\rho_{\rm soft}$, since we are now considering the case of  $r(\rho_{\rm soft})=0$, which means that $\rho_{m,j}< \rho_{\rm soft}$ holds for any $(m,j)\in[N_{\rm temp}]_0\times\left[M\right]_0$.
    Accordingly, the algorithm outputs ``there is no signal" at the second step. 
    In summary, at any rate, the algorithm outputs this message, if $r(\rho_{\rm soft})=0$.
    
    According to Theorem \ref{th:QAA}, in any cases, the number of calls to $O_{\rm AE}$ in $\proc{QAA}\left(O_{\rm AE},\frac{1}{2N_{\rm temp}M},\frac{\delta}{2}\right)$ is at most $O\left(\sqrt{N_{\rm temp}M}\log \delta^{-1}\right)$.
    As stated above, the number of calls to $O_{\rm Re}$ and $O_{\rm Im}$ in $O_{\rm AE}$ is given by Eq.~(\ref{eq:numOrho}).
    Combining these, we obtain the bound Eq.~(\ref{eq:compNoSigCase}) for the number of queries to $O_{\rm Re}$ and $O_{\rm Im}$ in the algorithm.
    
    Since $\proc{QAA}\left(O_{\rm AE},\frac{1}{2N_{\rm temp}M},\frac{\delta}{2}\right)$ runs on the same system in any cases, the discussion on qubit number is the same as Case (i).\\
    
    \noindent (iii) $r(\rho_{\rm hard})=0$ and $r(\rho_{\rm soft})>0$\\
    
    Since the evaluation on the maximum number of calls to $O_{\rm Re}$ and $O_{\rm Im}$ obtained in Case (ii) also applies to this case, we have the same query complexity bound of Eq.~(\ref{eq:compNoSigCase}).
    The discussion on qubit number is also the same as Case (ii).
    
\end{proof}

\bibliography{reference}
\bibliographystyle{unsrt}

\end{document}